\documentclass[12pt]{amsart}
\usepackage[margin=1in]{geometry}
\usepackage{amsfonts}
\usepackage{amssymb}
\usepackage{epsfig}
\usepackage{dsfont}
\usepackage{euscript}
\usepackage{color}
\usepackage{hyperref}
\usepackage{float}
\usepackage{stmaryrd}
\usepackage{verbatim}
\usepackage[normalem]{ulem}
\DeclareMathOperator*{\bigtimes}{\vartimes}
 \newtheorem{thm}{Theorem}[section]
 
 \newtheorem{lemma}[thm]{Lemma}
 \newtheorem{prop}[thm]{Proposition}
 \theoremstyle{definition}
 \newtheorem{defi}[thm]{Definition}
 \newtheorem*{property}{Property}
 \theoremstyle{remark}
 \newtheorem{rem}[thm]{Remark}

\newtheorem*{assumption}{Assumption}
 
 \numberwithin{equation}{section}

\newcommand{\eb}{\overline{\mathbb{E}}}
\newcommand{\e}{{\rm e}}

\newcommand{\rd}{{\rm d}}

\newcommand{\Hilb}{\mathcal{H}}

\newcommand{\bea}{\begin{eqnarray}}
\newcommand{\eea}{\end{eqnarray}}
\newcommand{\dd}{\mathrm{d}}
\newcommand{\ee}{\mathrm{e}}
\newcommand{\ii}{\mathrm{i}}
\renewcommand{\gg}{\mathfrak{g}}
\renewcommand{\sp}{\mathfrak{sp}}

\newcommand{\C}{\mathbb{C}}
\newcommand{\R}{\mathbb{R}}

\newcommand{\Z}{\mathbb{Z}}

\newcommand{\E}{{\mathbb E}}

\newcommand{\V}{{\mathbb V}}
\newcommand{\K}{{\mathbb K}}
\newcommand{\N}{\mathbb{N}}

\newcommand{\PP}{\mathbb{P}}

\newcommand{\M}{\mathbb{M}}

\newcommand{\DON}{D_0^{(N)}}
\newcommand{\DO}{D_{\omega}}
\newcommand{\DOsp}{D_{\omega,\ell}^{(N)}}
\newcommand{\DONxL}{D_{\omega}^{(x,L)}}
\newcommand{\DONL}{D_{\omega}^{(L)}}
\newcommand{\DONOL}{D_{\omega}^{(0,L)}}

\newcommand{\SO}{\mathrm{SO}}
\newcommand{\SN}{\mathrm{S}_{\mathrm{N}}(\R)\,}

\renewcommand{\Re}{\operatorname{Re}}
\renewcommand{\Im}{\operatorname{Im}}
\newcommand{\uu}{u^{\uparrow}}
\newcommand{\ud}{u_{\downarrow}}
\DeclareMathOperator{\tr}{tr}

\DeclareMathOperator{\diag}{diag}
\DeclareMathOperator{\supp}{supp}
\DeclareMathOperator{\s}{span}

\DeclareMathOperator{\dist}{dist}
\DeclareMathOperator{\GL}{GL}

\newcommand{\Dom}{\mathrm{Dom}}
\def\ec{{\mathbb E}}
\def\pc{{\mathbb P}}
\newcommand{\SpN}{\mathrm{Sp}_{\mathrm{N}}(\R)}
\newcommand{\spN}{\mathfrak{sp}_{\mathrm{N}}(\R)}
\newcommand{\spONN}{\mathfrak{spo}_{\mathrm{N}}(\R)}
\newcommand{\SpC}{\mathrm{Sp}_{\mathrm{N}}^*(\C)}
\newcommand{\SpONN}{\mathrm{SpO}_{\mathrm{N}}(\R)}

\newcommand{\dlO}{d_{\log\, \mathcal{O}}}
\newcommand{\omO}{\omega^{(0)}}

\definecolor{Plum}{rgb}{.5,0,1}

\numberwithin{equation}{section}

\def\d{\mathrm{d}}

\begin{document}

\title{Localization for random quasi-one-dimensional Dirac operators}

\author[H. Boumaza]{Hakim Boumaza*}
\address{LAGA, Universit\'e Sorbonne Paris Nord, 99 avenue J.B. Cl\'ement, F-93430 Villetaneuse}
\email{boumaza@math.univ-paris13.fr}
\thanks{*Corresponding author. \\ The research of H.B was partially funded by the ANR Project RAW, ANR-20-CE40-0012}
\author[S. Zalczer]{Sylvain Zalczer}
\address{KIT, Englerstrasse 2, 76131, Karlsruhe, Germany}
\email{sylvain.zalczer@kit.edu}
\thanks{\noindent The research of S.Z was funded by the Deutsche Forschungsgemeinschaft (DFG, German Research Foundation) - Project-ID 258734477 - SFB 1173.}

\date{}

%
\begin{abstract}
    We consider a random family of Dirac operators on  $N$ parallel real lines with an ergodic matrix-valued random potential. We establish a  criterion for Anderson and dynamical localization involving properties on the group generated by transfer matrices. In particular, we consider not only the usual case where this group is the symplectic group but also a strict subgroup of it. We establish under quite general hypotheses that the sum of the Lyapunov exponents and the integrated density of states are H\"older continuous. Moreover, for a set of concrete cases where the potentials are on Pauli matrices, we compute the transfer matrices and prove either localization or delocalization, depending on the potential and on the parity of $N$.
\end{abstract}

\maketitle

\section{Introduction} \label{sec_intro}

Dirac introduced the operator bearing his name in order to describe the motion of a relativistic electron. While nonrelativistic quantum matter is described by the Schr\"odinger operator, which is a second-order differential operator acting on scalar-valued functions, the Dirac operator has order 1 and acts on vector-valued functions. A standard reference on the Dirac operator from a mathematics point of view is the book of Thaller~\cite{thaller}.

In the last twenty years, the Dirac operator has been used as an effective Hamiltonian to describe graphene samples, although no relativistic effects are considered. Graphene is a 2-dimensional material made of carbon atoms arranged according to a honeycomb structure. The setting which we are mostly interested in this paper is the one of graphene nanoribbons, which consist in an infinite band of graphene, bounded transversely. In~\cite{BBKOB}, a model of Dirac operator on a waveguide is studied. 

The question of localization on graphene nanoribbons has already been considered by physicists, although they generally consider a discrete tight-binding model instead of the Dirac operator. In~\cite{dSP}, some localization is obtained for a quasi-periodic perturbation. In~\cite{NOR} and~\cite{EZXH}, different Anderson-like models are considered. 
In~\cite{Xiong}, several interesting phenomena are highlighted. First, different types of disorder can lead to different localization regimes. Second, in some cases there is localization for all nonzero energy but delocalization for energy zero, for symmetry reasons. Third, the Lyapunov exponents can be grouped into pairs, here again because of some symmetry of the system. We will be able to recover these properties in our results.

From a mathematical point of view, a few papers have already considered random Dirac operators. In~\cite{BCZ}, Barbaroux, Cornean and the second author proved localization at band edges for a gapped Dirac Hamiltonian in any dimension. The latter continued with~\cite{Zdos}, where he proved the Lipschitz regularity of the integrated density of states for the same model, and~\cite{Z23}, in which Anderson localization is proven at any  energy for a one-dimensional model. The present paper is an extension of~\cite{Z23} to the quasi-one-dimensional case.  Discretized versions of the Dirac operator have been studied by Prado and de Oliveira, either for a Bernoulli random potential~\cite{PdO07, PdOdO} or for a potential with an absolutely continuous distribution~\cite{PdOC}. In~\cite{SSB}, Sadel and Schulz-Baldes consider a quasi-one-dimensional random Dirac operator with some symmetry property. They look at the Lyapunov exponents and establish delocalization in some cases, but they do not prove localization.

In the present paper, we consider a semi-discretized model: the nanoribbon is represented by a finite number of parallel infinite lines. This makes it possible to use techniques which are specific to one-dimensional systems, based on transfer matrices and Lyapunov exponents. These methods have already been used to prove localization for Schr\"odinger operators  with random potentials even with singular distributions, in~\cite{KLS} in the discrete case and in~\cite{Bou09} in the continuous case. A survey article on these methods has recently been published~\cite{Bou23}. The proof of the localization itself relies on the method called \emph{multiscale analysis}, in its version developped by Germinet and Klein in~\cite{GKbootstrap}.

Our paper is organized as follows. In Section~\ref{sec_intro}, we present the model as well as the main tools used in the analysis, and we state our results. In Section~\ref{sec:loc-crit}, we prove a general criterion ensuring localization for a quasi-one-dimensional random Dirac operators and a second one valid in some particular case for which we cannot apply the general criterion. In Section~\ref{sec:splitting:Pauli}, we prove that we can apply thess critera to some explicit examples.  Finally, in Appendix~\ref{app:G5}, we study the properties of a group appearing in one of the cases considered in Section~\ref{sec:splitting:Pauli}.
\bigskip

\subsection{Quasi-one-dimensional operators of Dirac type}



Given an integer $N\geq 1$, the free Dirac operator on $N$ parallel straight lines is
\begin{equation}
    \DON:=J\frac{\dd}{\dd x}, \text{ with }J:=\begin{pmatrix}0&-I_N\\I_N&0\end{pmatrix}\text{ and } \Dom( \DON)=H^1(\R)\otimes \C^{2N}.
\end{equation}
It is easy to see that this operator is self-adjoint.

We add to this free operator a random potential. Let $(\Omega, \mathcal{A},\mathbb{P})$ a complete probability space and $\ell>0$ be a disorder parameter: the smaller $\ell$ is, the "more disordered" the system is. The random potential  $(V_{\omega}^{(n)})_{n\in \Z}$ is a sequence of  independent and identically distributed (i.i.d. for short) random variables such that, for every $n\in \Z$, the function  $x\mapsto V_{\omega}^{(n)} (x)$ takes values in the Hermitian matrices, is supported in $[0,\ell]$ and is  uniformly bounded in $x$, $n$ and $\omega$.

We consider the random family $\{\DO \}_{\omega \in \Omega}$ of quasi-one-dimensional Dirac operators defined for every realization $\omega\in \Omega$ by:
\begin{equation}\label{def_op_D_general}
\DO : = \DON +   \sum_{n\in \Z} V_{\omega}^{(n)} (\cdot-\ell n).
\end{equation}
 Under such conditions, for each $\omega \in \Omega$, the operator $\DO $ is self-adjoint on the Sobolev space  $H^1(\R)\otimes \C^{2N}$ and thus, for every $\omega\in \Omega$, the spectrum of $\DO $, denoted by $\sigma\left(\DO \right)$, is included in $\R$.
 
 The random potential is such that $\{\DO \}_{\omega \in \Omega}$ is a $\ell \Z$-ergodic  random family of operators. As a consequence,  there exists $\Sigma \subset \R$ such that, for $\pc$-almost every $\omega \in \Omega$, $\Sigma=\sigma(\DO )$. There also exist $\Sigma_{\mathrm{pp}}$, $\Sigma_{\mathrm{ac}}$ and $\Sigma_{\mathrm{sc}}$, subsets of $\R$, such that, for $\mathbb{P}$-almost every $\omega \in \Omega$, $\Sigma_{\mathrm{pp}}=\sigma_{\mathrm{pp}}(\DO)$, $\Sigma_{\mathrm{ac}}=\sigma_{\mathrm{ac}}(\DO)$ and $\Sigma_{\mathrm{sc}}=\sigma_{\mathrm{sc}}(\DO)$, respectively the pure point, absolutely continuous and singular continuous spectrum of $\DO$ (see for example~\cite[Theorem~4.3]{Kirsch-notes}).  


We aim at proving that under conditions involving notions coming from the theory of dynamical systems, the phenomenon of Anderson localization occurs for $\{\DO \}_{\omega \in \Omega}$  at all energies, except maybe those in a discrete set.
There are several mathematical definitions to translate the Anderson localization phenomenon for a general family of random operators. Let $\{ H_{\omega} \}_{\omega \in \Omega}$  be a family of self-adjoint random operator on a Hilbert space, which will be the space $H^1(\R)\otimes \C^{2N}$ in our model \eqref{def_op_D_general}.

\begin{defi}[AL]\label{def_anderson_loc}
Let $I$ be an interval of $\R$. We say that the family $\{ H_{\omega} \}_{\omega \in \Omega}$ of almost-sure spectrum $\Sigma$ has the property of \emph{Anderson localization} in $I$ when:
\begin{enumerate}
\item $\Sigma \cap I = \Sigma_{\mathrm{pp}} \cap I\neq \emptyset$ and $\Sigma_{\mathrm{ac}}\cap I=\Sigma_{\mathrm{sc}}\cap I=\emptyset$,
\item the eigenfunctions associated with the eigenvalues in $\Sigma \cap I$ decay exponentially to $0$ at infinity. 
\end{enumerate}
\end{defi}

Note that if $\{ H_{\omega} \}_{\omega \in \Omega}$ exhibits Anderson localization in $I$, $\pc$-almost surely the point spectrum of $H_{\omega}$ is dense in $\Sigma \cap I$. The definition of Anderson localization is a stationary definition, involving only the Hamiltonian $H_{\omega}$ and not the associated one-parameter group. The following definition takes into account the dynamics in time of the wave packets.

\begin{defi}[DL]\label{def_dyn_loc}
  Let $I$ be an interval of $\R$. We say that the family $\{ H_{\omega} \}_{\omega \in \Omega}$  is \emph{dynamically localized} in $I$, when
\begin{enumerate}
 \item   $\Sigma \cap I \neq \emptyset$,
\item for every compact interval $I_0\subset I$, every $\psi\in L^2(\R)\otimes \C^{2N}$ with compact support and every $p\geq 0$,
\begin{equation}\label{eq_def_loc_dynamique_continu}
\E\left( \sup_{t\in \R} \left\|(1+|\cdot|^2)^{\frac{p}{2}} \ee^{-\mathrm{i}tH_{\omega}}  \mathbf{1}_{I_0}(H_{\omega}) \psi \right\|_{L^2(\R)}^2 \right) < +\infty
\end{equation}
where $\mathbf{1}_{I_0}(H_{\omega})$ denotes the spectral projector on $I_0$ associated with $H_{\omega}$.
\end{enumerate}
\end{defi}

\noindent  The definition \ref{def_dyn_loc} is dynamical in nature and follows the evolution of wave packets over time. It tells us that the solutions of the Schr\"odinger equation are localized in space in the vicinity of their initial position and this, uniformly over time. This reflects the absence of quantum transport. More precisely, let $|x|$ be the position operator, \emph{i.e.}, the multiplication operator by $x\mapsto (1+||x||_2^2)^{\frac12}$ on $L^2(\R)\otimes \C^{2N}$. For any state $\psi \in L^2(\R)\otimes \C^{2N}$, if we denote $\psi_I(t)$ as the evolution of the spectral projection of $\psi$ at time $t$, then the moments of the position operator are bounded in $t$:
$$\exists C_{\psi,I}>0,\ \forall p\geq 0,\ \forall t\in \R,\ \left\langle\psi_I(t), |x|^p \psi_I(t)\right\rangle \leq C_{\psi,I}.$$ 
On the contrary, we will say that there is quantum transport in an interval $I'$ when
$$
\exists  \alpha >0,\ \exists p\geq 0,\ \forall t\in \R,\ \langle\psi_{I'}(t), |x|^p \psi_{I'}(t)\rangle \geq |t|^{\alpha}.
$$

\noindent Let us point out that the use of multiscale analysis as done in this paper will imply both Anderson localization and dynamical localization \cite{DS01}. Note that dynamical localization implies absence of continuous spectrum but does not imply in general, the exponential decay of the eigenfunctions as in (AL). It is also possible to define even stronger notions of localization, all of them being implied by the use of multiscale analysis. For an exhaustive presentation of these notions we refer to the third part of \cite{Kl08}.


Before stating our main results, two localization criteria for quasi-one-dimensional operators of Dirac type, we need to introduce the Lyapunov exponents and the Furstenberg  group of such operators.

\subsection{Transfer matrices, Lyapunov exponents and the Furstenberg  group}\label{sec:trans_mat}

In order to determine the almost-sure spectrum of $\{\DO \}_{\omega \in \Omega}$ and to study the asymptotic behaviour of the corresponding generalized eigenfunctions, one considers the equation for the generalized eigenvalues, for every $\omega \in \Omega$,
\begin{equation}\label{eq_eigenfunctions}
\DO u = E u,\quad \mbox{ where }  E\in \C \ \mbox{  and  }\ u=\left( \begin{smallmatrix}                                                 u^{\uparrow}\\ u_{\downarrow}                                                \end{smallmatrix}\right)\ :\ \R \to \C^{2N} \ .
\end{equation}
The notation $ u=\left( \begin{smallmatrix}                                                 u^{\uparrow}\\ u_{\downarrow} \end{smallmatrix} \right)$ refers to the decomposition spin up / spin down of the solution of the Dirac equation. 

Equation (\ref{eq_eigenfunctions}) is a linear differential system of order $1$. We introduce, for $E\in \C$ and every $x,y \in \R$, the transfer matrix  $T_x^y(E)$ of $\DO$ from $x$ to $y$ which maps a solution $(\uu,\ud)$  at time $x$ to the same solution at time $y$. It is defined by the relation 
\begin{equation}\label{eq_def_transfer_mat_xy}
\left( \begin{array}{c}
\uu(y) \\
\ud(y) 
\end{array} \right) = T_x^y(E) \left( \begin{array}{c}
\uu(x) \\
\ud(x) 
\end{array} \right)
\end{equation}
and in particular, $T_x^x(E)=I_{\mathrm{2N}}$ for every $x\in \R$. 
The transfer matrices are elements of the complex symplectic group
 \begin{equation}\label{eq:def_SpC}
 \SpC=\{M \in \mathcal{M}_{\mathrm{2N}}(\C)\ |\ M^*JM=J \}
 \end{equation}
 with $J=\left( \begin{smallmatrix} 0 & -I_N\\ I_N & 0  \end{smallmatrix} \right)$. Indeed, for $E\in \C$ and $x\in \R$ fixed  $y \mapsto T_x^y(E)$ satisfies $\DO  (T_x^y(E)) = E (T_x^y(E)) $ on $\R$. It implies $\left(\frac{d}{dy} T_x^y(E)\right)^* J T_x^y(E) + (T_x^y(E))^* J \frac{d}{dy} T_x^y(E) = 0 $. Hence, the function  $y\mapsto (T_x^y(E))^* J  T_x^y(E)$ is constant on $\R$. Taking the value at $y=x$ one obtains $J$ and $(T_x^y(E))^* J  T_x^y(E) =J$ for every $y \in \R$.

\begin{rem}
Note that despite its name, $\SpC$ is a real Lie group since it is a $C^{\infty}$ manifold and not a holomorphic manifold because of the presence of a conjugation in its definition.
\end{rem}
 
For $E\in \C$ fixed and two couples $(x,y)$ and $(x',y')$ in $\R^2$, the random matrices $T_x^y(E)$ and $T_{x'}^{y'}(E)$ are not necessarily independent. In order to apply the results of the theory of products of i.i.d. random matrices, we also introduce, for every $n\in \Z$, the transfer matrices $T_{\omega^{(n)}}(E)=T_{\ell n}^{\ell(n+1)}(E)$ from $\ell n$ to $\ell (n+1)$. The transfer matrix  $T_{\omega^{(n)}}(E)$ is therefore defined by the relation
\begin{equation}\label{eq_def_transfer_mat}
\left( \begin{array}{c}
\uu(\ell(n+1)) \\
\ud(\ell(n+1)) 
\end{array} \right) = T_{\omega^{(n)}}(E) \left( \begin{array}{c}
\uu(\ell n) \\
\ud(\ell n) 
\end{array} \right)
\end{equation}
for all $n\in \Z$. 

The sequence $(T_{\omega^{(n)}}(E))_{n\in \Z}$ is  a sequence of i.i.d.\ matrices because of the i.i.d.\ character of the $V_{\omega}^{(n)}$'s and the disjointness of their supports for different values of $n$.

By iterating the relation (\ref{eq_def_transfer_mat}) we get the asymptotic behaviour of $\left( \begin{smallmatrix}      
u^{\uparrow}\\ u_{\downarrow}         
\end{smallmatrix}\right)$. More precisely, we introduce, for $E\in \C$ fixed, the cocycle $\Phi_E : \Z\times \Omega \to \SpC$ defined by : $\forall n\in \Z,\ \forall \omega \in \Omega,$
$$ \Phi_E(n,\omega) \;=\; \left\lbrace 
\begin{array}{lcl}
T_{\omega^{(n-1)}}(E)\cdots T_{\omega^{(0)}}(E) & \mbox{ if } & n >0 \\
I_{\mathrm{N}} & \mbox{ if } & n=0 \\
(T_{\omega^{(n)}}(E))^{-1} \cdots (T_{\omega^{(-1)}}(E))^{-1} & \mbox{ if } & n <0
\end{array} \right.$$

To get the \emph{exponential} asymptotic behaviour of $\left( \begin{smallmatrix}                                                 u^{\uparrow}\\ u_{\downarrow}                                                \end{smallmatrix}\right)$, we define the exponential growth (or decay) exponents of the product of random matrices $T_{\omega^{(n-1)}}(E)\cdots T_{\omega^{(0)}}(E)$.

\begin{defi}\label{def_lyap_exp}
Let $E\in \C$. The \emph{Lyapunov exponents} $\gamma_{1}^{\pm}(E),\ldots,\gamma_{2N}^{\pm}(E)$ associated with the sequence $(T_{\omega^{(n)}}(E))_{n\in \Z}$ are defined inductively by
\begin{equation}\label{eq_def_lyap_exp}
\sum_{i=1}^{p} \gamma_{i}^{\pm}(E) = \lim_{n \to \pm\infty} \frac{1}{|n|}
\mathbb{E}(\log ||\wedge^{p} \Phi_E(n,\omega) ||)
\end{equation}
for every $p\in \{1,\ldots,2N\}$. Here, $\wedge^{p} M$ denotes the $p$-th exterior power of the matrix $M$, acting on the $p$-th exterior power of $\C^{2N}$.
\end{defi}

Since the transfer matrices all lie in $\SpC$, for every $i\in \{1,\ldots , 2N \}$, $\gamma_{i}^{+}(E)=\gamma_{i}^{-}(E)$. Indeed, for each $M\in \SpC$, $||M||=||M^{-1}||$.

One can link the Lyapunov exponents to the singulars values of the cocycle $\Phi_E(n,\omega)$.
\begin{prop}[\cite{BL},Proposition A.III.5.6]\label{prop_lyapvps}
If $s_{1}(\Phi_E(n,\omega))\geq \ldots \geq s_{2N}(\Phi_E(n,\omega))>0$ are the singular values of $\Phi_E(n,\omega)$, then, for $\mathsf{P}$-almost every $\omega \in \Omega$,
$$\forall p\in \{1,\ldots, 2N\},\ \gamma_{p}^{\pm}(E)= \lim_{n \to \pm \infty} \frac{1}{|n|} \E(\log s_{p}(\Phi_E(n,\omega)) = \lim_{n \to \pm \infty} \frac{1}{|n|} \log s_{p}(\Phi_E(n,\omega)).$$
\end{prop}

This implies in particular that $\gamma_{1}(E)\geq \ldots \geq \gamma_{2N}(E)$. Moreover, the symplecticity of the transfer matrices also implies the following symmetry property (\emph{cf.}~\cite[Proposition~A.IV.3.2]{BL}) 
$$\forall i \in \{1,\ldots,N\},\ \gamma_{2N-i+1}(E)= -\gamma_{i}(E).$$

\bigskip

\noindent To study the properties of the Lyapunov exponents, we introduce the group which contains all the different products of transfer matrices, the so-called Furstenberg  group. 

\begin{defi}\label{def_GE}
For every $E\in \C$, let $\mu_E$ be the common distribution of the random matrices $T_{\omega^{(n)}}(E)$.  We define the \emph{Furstenberg  group} of $\{\DO \}_{\omega \in \Omega}$ at $E$ as the closed group generated by the support of $\mu_E$,  
$$G(E)=\overline{<\supp \mu_E>},$$
where   the closure is taken for the usual topology in $\mathcal{M}_{2N}(\C)$.
\end{defi}

We already remark that for all $E\in \C$, $G(E)$ is a subgroup of $\SpC$.

\subsection{Localization criteria for quasi-one-dimensional operators of Dirac type}\label{sec:localization_criterion}

The formalism of transfer matrices, Lyapunov exponents and the Furstenberg  group enables to state criteria of dynamical localization for quasi-one-dimensional operators of Dirac type. 

Before that, we introduce several definitions in order to fix the framework in which we are able to obtain such criteria of dynamical localization. 

We introduce two properties concerning the Furstenberg group. Let $p\in \{1,\ldots, N\}$.  The first property is called \emph{$p$-contractivity}.

\begin{defi}
A subset $T$ of $\GL_{\mathrm{2N}}(\C)$ is called \emph{p-contracting}  if there exists a sequence $(M_n)$ of elements of $T$ such that $\|\Lambda^p M_n\|^{-1}\Lambda^p M_n$ converges to a rank-one matrix.
\end{defi}

Let $L\geq 1$ an integer. For $l\in \{1,\ldots, L\}$ we denote indifferently by $b_l$ a bilinear form on $\C^{2N}$ or its matrix in the canonical basis of $\C^{2N}$. We also denote by $b_0$, or most simply by $J$,  the symplectic bilinear form on $\C^{2N}$ associated with the matrix $J=\left( \begin{smallmatrix} 0 & -\mathrm{I}_{\mathrm{N}}\\ \mathrm{I}_{\mathrm{N}} & 0  \end{smallmatrix} \right)$.

For any $p\in \{1,\ldots, N\}$, let $(J,b_1,\ldots,b_L)$-$L_p$ be the vector subspace of $\Lambda^p \C^{2N}$ whose elements are $p$-decomposable vectors $u_1\wedge \cdots \wedge u_p$ such that
$$\forall i,j\in  \{1,\ldots, p\},\ \forall l\in  \{0,\ldots, L\},\ b_l(u_i,u_j)=0$$
\emph{i.e}, the family $(u_1,\ldots, u_p)$ is orthogonal for all the bilinear forms $b_l$ (and in particular each $u_i$ is orthogonal to itself for all $b_l$).

The second property is called the $(J,b_1,\ldots,b_L)$-$L_p$-strong irreducibility. It generalizes the notion of $L_p$-strong irreducibility as defined in \cite{BL} in the setting of the real symplectic group.

\begin{defi}
We say that a subset $T$ of $\GL_{\mathrm{2N}}(\C)$  is \emph{$(J,b_1,\ldots,b_L)$-$L_{p}$-strongly irreducible} if there does not exist any $W$, finite union of proper vector subspaces of $(J,b_1,\ldots,b_L)$-$L_{p}$, such that $(\Lambda^{p}M)(W)=W$ for all $M$ in $T$. 
\end{defi}


 
We now state two localization criteria for $\{ \DO \}_{\omega \in \Omega}$ in some particular cases of the group $G$. The first one states for the group $G=\SpC$.

\begin{thm}\label{thm:loc_criterion1}
We fix a compact interval $I\subset \R$. We assume that there exists an open interval $\tilde{I}$ containing $I$ and such that for every $E\in \tilde{I}$:
\begin{enumerate}
    \item the Furstenberg group $G(E)$ is included in $\SpC$  ;
    \item for every $p\in \{1,\ldots, N\}$, $G(E)$ is $p$-contracting and $J$-$L_p$-strongly irreducible.
\end{enumerate}
Then $\{\DO \}_{\omega \in \Omega}$ exhibits dynamical localization in $\Sigma \cap I$.
\end{thm}

In order to state a second theorem, we introduce the matrix
\begin{equation}\label{def:S}    
S:=\left( \begin{array}{cc}
   K  &  0\\
   0  & K
\end{array}\right)\in \mathcal{M}_{2N}(\R),
\end{equation}
where $K$ is the diagonal matrix with $(-1)^{i+1}$ in position $i$. We define the group
\begin{equation}\label{def:SpONN}
 \SpONN:=\left\{M\in\mathcal{M}_{2N}(\R),\ ^tMJM=J,\ ^t MSM=S\right\}.
\end{equation}

The second localization criterion states for the group $G=\SpONN$.

\begin{thm}\label{thm:loc_criterion2}
Assume that $N$ is even. We fix a compact interval $I\subset \R$. We assume that there exists an open interval $\tilde{I}$ containing $I$ and such that for every $E\in \tilde{I}$:
\begin{enumerate}
    \item the Furstenberg group $G(E)$ is included in $\SpONN$  ;
    \item for every $2p\in \{1,\ldots, N\}$, $G(E)$ is $2p$-contracting and $(J,S)$-$L_{2p}$-strongly irreducible.
\end{enumerate}
Then $\{\DO \}_{\omega \in \Omega}$ exhibits dynamical localization in $\Sigma \cap I$.
\end{thm}

Theorems \ref{thm:loc_criterion1} and \ref{thm:loc_criterion2} are comparable to \cite[Theorem 1]{Bou09} which is a criterion of localization for quasi-one-dimensional operators of Schr\"odinger type. This former result deals only with the case of $G=\SpN$.
\bigskip

The proofs of Theorems \ref{thm:loc_criterion1} and \ref{thm:loc_criterion2}  involve several steps as detailed in Section \ref{sec:loc-crit}:
\begin{enumerate}
  \item The assumptions of the two theorems lead to an integral formula for the Lyapunov exponents which implies their H\"older regularity.
\item We then deduce the same H\"older regularity for the integrated density of states, using a Thouless formula.
\item From this regularity of the  integrated density of states, we get a weak Wegner's estimate adapted to Bernoulli randomness.
\item Finally we apply a multiscale analysis scheme which involves the proof of an Initial Length Scale Estimate.
\end{enumerate}

Actually, we will see in Section \ref{sec:loc-crit} that most of these steps except the proof of the  Initial Length Scale Estimate in the last one are true with more general hypothesis. Let us state them now.

\begin{assumption}[$\mathbf{L^{(N)}}$]\label{assumptions}

 We fix a compact interval $I\subset \R$. We assume that there exists $L_N$ a vector subspace of  $\Lambda^N\C^{2N}$ and an open interval $\tilde{I}$ containing $I$ such that, for all $E\in \tilde{I}$:
\begin{itemize}
  
 \item[$\mathbf{(L_1^{(N)}})$] for all $g\in G(E)$, $(\Lambda^N g)(L_N)\subset L_N$;
   
\item[$\mathbf{(L_2^{(N)}})$] for all $x\neq0$ in $L_{N}$, 
$$ \lim_{n\to\infty}\frac1n\E(\log\|\Lambda^{N} \Phi_E(n,\cdot) x\|)=\sum_{i=1}^{N}\gamma_i(E);$$
\item[$\mathbf{(L_3^{(N)}})$] there exists a unique probability measure $\nu_{N,E}$ on $\mathsf{P}(L_{N})$ which is $\mu_E$-invariant, such that
$$\gamma_1(E)+\dots+\gamma_{N}(E)=\int_{G(E)\times \mathsf{P}(L_{N})}\log\frac{\|(\Lambda^{N}g)x\|}{\|x\|}\rd \mu_E(g)\rd \nu_{N,E}(\bar{x});$$
\item[$\mathbf{(L_4^{(N)}})$] $\gamma_1(E)+\dots+\gamma_{N}(E)>0$.

\end{itemize}   
\end{assumption}

The properties of $N$-contractivity and $J$ or $(J,S)$-$L_N$-strong irreducibility imply Assumption $\mathbf{(L^{(N)}})$. It is a consequence of \cite[Proposition A.IV.3.4]{BL} in the case of $\SpC$ since one has the identification between $\SpC$ and $\mathrm{Sp}_{2\mathrm{N}}(\R)$  using the following application which split the real and imaginery parts of the matrices in $\mathcal{M}_{\mathrm{2N}}(\C)$: 
$$\pi\ :\ \begin{array}{cll}
           \mathcal{M}_{\mathrm{2N}}(\C) & \to & \mathcal{M}_{\mathrm{4N}}(\R) \\[2mm]
	    A+iB & \mapsto & \left( \begin{smallmatrix}
				      A & -B \\
				      B & A
				      \end{smallmatrix} \right) .
				\end{array}$$
In the case of $\SpONN$, it is proven in Proposition \ref{prop:sepLyap}.

In the following Section, we will present some explicit cases of the model $\{ \DO \}_{\omega \in \Omega}$ for which we are able to verify the assumptions of Theorem \ref{thm:loc_criterion1} or Theorem \ref{thm:loc_criterion2}.

\subsection{Application of the localization criteria to a class of splitting potentials}
\label{sec:exemples}

In this Section, we introduce a particular case of quasi-one-dimensional operators of Dirac type whose potentials split in a sum of two tensorized Pauli matrices.
 Recall the usual notations for Pauli matrices :
 $$\sigma_0 :=  \mathrm{I}_2,\ \sigma_1: = \left( \begin{smallmatrix}
 0 & 1 \\
 1 & 0   \end{smallmatrix} \right),\ \sigma_2: =  \left( \begin{smallmatrix}
 0 & -\ii \\
 \ii & 0                                                                           \end{smallmatrix} \right),\ \sigma_3 := \left( \begin{smallmatrix}
 1 & 0 \\
 0 & -1  \end{smallmatrix} \right).$$


Consider the particular family $\{\DOsp \}_{\omega \in \Omega}$ where the potential split into a periodic part and a random part :
\begin{equation}\label{def_op_D}
 \DOsp: = \DON + V_{\mathrm{per}} + V_{\omega}.
\end{equation}

The potential $ V_{\mathrm{per}}$ is a $\ell$-periodic function, linear combination of tensorized Pauli matrices  of the form
\begin{equation}\label{decom:Vper}V_{\mathrm{per}}: = \left(\alpha_0 \sigma_0 +  \alpha_1 \sigma_1 +\alpha_2 \sigma_2   + \alpha_3 \sigma_3 \right)  \otimes \hat{V}_{\mathrm{per}},\end{equation}
\noindent where $\alpha_0,\ldots, \alpha_3$ are real numbers and $\hat{V}_{\mathrm{per}}$ is a $\ell$-periodic function with value in the space of the $N$-by-$N$ real symmetric matrices denoted by $\SN$. Note that : 
$$\sigma_0 \otimes V: = 
\left( \begin{smallmatrix} 
V & 0 \\
0 & V
\end{smallmatrix} \right)\in \mathcal{M}_{\mathrm{2N}}(\C)$$
and the same for the other tensor products.

We construct the random potential $ V_{\omega}$ in the following way.
Let $(\tilde{\Omega}_1, \tilde{\mathcal{A}}_1,\tilde{\mathbb{P}}_1),\ldots, (\tilde{\Omega}_N, \tilde{\mathcal{A}}_N,\tilde{\mathbb{P}}_N)$ be $N$ complete probability spaces. We take
$$(\Omega, \mathcal{A},\mathbb{P})=\left( \bigtimes_{n\in \Z} \tilde{\Omega}_1 \times \cdots \times \tilde{\Omega}_N , \bigotimes_{n\in \Z} \tilde{\mathcal{A}}_1 \otimes \cdots \otimes \tilde{\mathcal{A}}_N, \bigotimes_{n\in \Z} \tilde{\mathbb{P}}_1 \otimes \cdots \otimes \tilde{\mathbb{P}}_N \right).$$
For $i\in \{1,\ldots ,N\}$, the sequences $(\omega_i^{(n)})_{n\in \Z}$ are independent of each other and each one is a sequence of i.i.d. real-valued  random variables on $(\widetilde{\Omega}_i,\widetilde{\mathcal{A}}_i, \widetilde{\mathbb{P}}_i)$. Let $\nu_i$ be the common law of the $\omega_i^{(n)}$'s. We assume that $\{0,1\} \subset \supp \nu_i$ and $\supp \nu_i$ is bounded. In particular, the $\omega_i^{(n)}$'s can be Bernoulli random variables which is the most difficult case of randomness to deal with since it will imply the smallest possible Furstenberg group. We also set, for every $n\in \Z$, {\small $\omega^{(n)}=(\omega_1^{(n)},\ldots,\omega_N^{(n)})$}, which is a random variable on {\small $(\widetilde{\Omega}_1 \times \cdots \times \widetilde{\Omega}_N, \widetilde{\mathcal{A}}_1  \otimes \cdots \otimes \widetilde{\mathcal{A}}_N,\widetilde{\mathbb{P}}_1 \otimes \cdots \otimes \widetilde{\mathbb{P}}_N)$} of law $\nu = \nu_1\otimes \cdots \otimes \nu_N$. 
 For each $n$ in $\Z$, we introduce the random function
\begin{equation*}
     \hat{V}_{\omega^{(n)}}:=\left(
\begin{smallmatrix}
 \omega_{1}^{(n)} v_1(\cdot -\ell n) & & 0\\ 
 & \ddots &  \\
0 & &  \omega_{N}^{(n)} v_N(\cdot -\ell n)\\ 
\end{smallmatrix}\right),
\end{equation*}
where the $v_i$ are measurable functions from $[0,\ell)$ to $\R$.
We take
\begin{equation}\label{decomp:Vomega}V_{\omega} =\left( \beta_0 \sigma_0 + \beta_1 \sigma_1  + \beta_2 \sigma_2 + \beta_3 \sigma_3  \right) \otimes \sum_{n\in\Z}\hat{V}_{\omega^{(n)}}\end{equation}
where $\beta_0,\ldots, \beta_3$ are real numbers.


In order to compute the Furstenberg  group associated with $\{ \DO \}_{\omega \in \Omega}$, we will  express the transfer matrices as matrix exponentials (\emph{cf.} Section~\ref{Liealg}). This is possible only when the potentials are constant on each interval $(n\ell, (n+1)\ell)$. For this reason, we only consider the case where $\hat{V}_{\mathrm{per}}$ is a constant function and $v_1=\dots=v_N=\mathbf{1}_{[0,\ell]}$. By a small abuse of notation, we will denote by $\hat{V}_{\mathrm{per}}$ (resp. $\hat{V}_{\omega^{(n)}}$) the unique value of the function $\hat{V}_{\mathrm{per}}$ (resp. $\hat{V}_{\omega^{(n)}}$).

Moreover, we will restrict ourselves to particular combinations of non-vanishing $\alpha_i$'s and $\beta_j$'s. For simplicity we will only consider real-valued potentials $V_{\mathrm{per}}$ and $V_{\omega}$, which corresponds to the absence of magnetic field and to the following assumption.

\begin{assumption}\label{assumption1}
We assume that $\alpha_2=\beta_2 =0$.
\end{assumption}

Hence we only consider potentials which are on $\sigma_0$, $\sigma_1$ and $\sigma_3$ which implies in particular that the corresponding Furstenberg  groups will be included in $\SpN$ instead of $\SpC$. Next, we consider only splitting potentials with one deterministic term and one random term which allows to reduce the number of cases in which we should compute the Furstenberg  group from 43 to 9.

\begin{assumption}\label{assumption2}
We assume that one and only one among $\alpha_0$, $\alpha_1$ and $\alpha_3$ is different from zero and one and only one among $\beta_0$, $\beta_1$ and $\beta_3$ is different from zero.
\end{assumption}

Since we have to choose one random potential and one deterministic one, there are \emph{a priori} nine possibilities. Nevertheless, it is possible to reduce this number to five in the following way.
If one sets for $(V_0,V_1,V_3)\in (\mathcal{M}_{\mathrm{N}}(\R))^3$, 
$$D(V_0,V_1,V_3)=D_0^{(N)}+\sigma_0\otimes V_0+\sigma_1\otimes V_1 + \sigma_3 \otimes V_3$$
 acting on  $H^1(\R)\otimes \R^2$, then 
 $$\forall (V_0,V_1,V_3)\in (\mathcal{M}_{\mathrm{N}}(\R))^3,\ D(V_0,V_1,V_3)=P(-D(-V_0,-V_3,-V_1))P^*$$
 with $P$ the unitary matrix defined by 
\begin{equation}\label{eq_def_P_unitary}
 P=\frac{1}{\sqrt{2}} \left( \begin{matrix}
       I_{\mathrm{N}} & I_{\mathrm{N}}\\ I_{\mathrm{N}} & -I_{\mathrm{N}}
     \end{matrix} \right).
\end{equation}


Hence, for any  $(V_0,V_1,V_3)\in (\mathcal{M}_{\mathrm{N}}(\R))^3$, the operators $D(V_0,V_1,V_3)$ and  $-D(-V_0,-V_3,-V_1)$ have the same spectrum and also the same pure point, absolutely continuous and singular continuous spectra. These two operators have transfer matrices which are unitarily equivalent (through $P$ defined at \eqref{eq_def_P_unitary}), hence their Lyapunov exponents are equal. They also have Furstenberg  groups which are unitarily equivalent (again through $P$) and there is localization for $D(V_0,V_1,V_3)$ if and only if there is localization for $-D(-V_0,-V_3,-V_1)$.
 As a consequence, there are five cases, as explained in the following table.
 

\begin{table}[H]
    \centering
        \begin{tabular}{c|c|c|c|}
 $V_{\omega}$ $\backslash$ $V_{per}$&$\sigma_0$&$\sigma_1$&$\sigma_3$\\
 \hline
 $\sigma_0$&1&5&5\\\hline
 $\sigma_1$&3&2&4\\\hline
 $\sigma_3$&3&4&2\\\hline
     \end{tabular}
    \caption{The five possible cases}
    \label{tab:cas}
\end{table}


For each of these cases, we  prove either localization or delocalization.

 \begin{thm}\label{thm:localization:exemple}
Denote by $V$ the unique value of the function $\hat{V}_{\mathrm{per}}$. For almost-every real symmetric matrix $V\in \SN$, there exist a finite set $\mathcal{S}_V\subset \R$ and $\ell_C:=\ell_C(N,V) >0$ such that, for every $\ell\in (0,\ell_C)$, there exists a compact interval $I(N,V,\ell)\subset \R$ such that if $I\subset I(N,V,\ell)\setminus \mathcal{S}_V$ is an open interval with $\Sigma \cap I \neq \emptyset$, then :
\begin{itemize}
\item[(i)] in case $1$, $\Sigma \cap I$ is purely a.c.;
\item[(ii)] in cases $2$, $3$ and $4$,  $\{ \DOsp \}_{\omega \in \Omega}$ exhibits  Anderson and dynamical localization in $\Sigma \cap I$;
\item[(iii)] in case $5$,  $\{ \DOsp \}_{\omega \in \Omega}$ exhibits Anderson and dynamical localization in $\Sigma \cap I$ if $N$ is even and there is presence of a.c. spectrum of multiplicity $2$  if $N$ is odd.
\end{itemize}
\end{thm}

Remark that from its construction given at Section \ref{Liealg}, the interval $I(N,V,\ell)$ tends to the whole real line $\R$ when $\ell$ tends to $0$.

The localization results come from applications of our localization criteria in Theorems~\ref{thm:loc_criterion1} and~\ref{thm:loc_criterion2}.
The presence of a.c. spectrum in case $1$ and in case $5$ when $N$ is odd
is a consequence of the following theorem of Sadel and Schulz-Baldes.
\begin{thm}[\cite{SSB}, Theorem~4]\label{thm:deloc_crit}
    Let $\{\DO \}_{\omega \in \Omega}$ be as in~\eqref{def_op_D_general}. Then, for $k\in\{1,\dots,N\}$, the set
    \[S_k:=\{E\in \R, \text{ exactly } 2k\text{ Lyapunov exponents vanish at }E\}\]
    is an essential support of the almost-sure absolutely continuous spectrum of multiplicity $2k$.
\end{thm}

This result of Kotani's theory for quasi-one-dimensional operators of Dirac type has to be seen as  a \emph{delocalization} result for $\{ \DO \}_{\omega \in \Omega}$.

\bigskip
\bigskip

\section{Proof of the localization criteria}\label{sec:loc-crit}

The proofs of Theorem~\ref{thm:loc_criterion1} and Theorem~\ref{thm:loc_criterion2}  consist of several steps, following the strategy already used in~\cite{KLS, CL90, Bou08,Bou09}.

The first step is to prove that, under the hypotheses of the theorems, Assumption~$\mathbf{(L^{(N)})}$ or $\mathbf{(L^{(N-1)})}$ is satisfied depending on the parity of $N$. As we  mentioned in Section~\ref{sec:loc-crit}, this is already known for a group which is $J$-$L_p$-strongly irreducible. We provide a proof in the case of a $(J,S)$-$L_p$-strongly irreducible group.

Then, we prove  that, under Assumption~$\mathbf{(L^{(N)})}$ or $\mathbf{(L^{(N-1)})}$ , the sum of all nonnegative Lyapunov exponents is H\"older continuous.
\begin{thm}\label{thm:regu_Lyap}
    Let $I$ be a compact interval on which Assumption~$\mathbf{(L^{(N)})}$ holds or on which Assumption~$\mathbf{(L^{(N-1)})}$ holds and $\gamma_N$ is identically $0$.
    Then, the sum of the Lyapunov exponents associated with $\{ \DO \}_{\omega \in \Omega}$ is H\"older continuous on $I$, \emph{i.e.} there exist two real numbers $\alpha>0$ and $C>0$ such that
    \begin{equation}\label{HoldLyap}
      \forall E, E'\in I, \left|\sum_{p=1}^N\left(\gamma_p(E)-\gamma_p(E')\right)\right|\leq C|E-E'|^\alpha.
    \end{equation}
\end{thm}
We prove this first result in Section~\ref{sec:reguLyap}. Note that~\eqref{HoldLyap} holds not only on intervals where the hypotheses of Theorem~\ref{thm:regu_Lyap} are satisfied but also trivially on intervals where $\sum_{p=1}^N\gamma_p$ is identically 0. A singularity can only occur at the boundary between such intervals.

As in \cite{Bou09,Z23}, the next step towards localization is to prove H\"older regularity of the integrated density of states. We get this regularity from the one of the Lyapunov exponents. We define the density of states in the same way as in \cite{Z23}.

Given $x\in\R$ and $L>0$, we define the operator $\DONxL$, called the \emph{restricted operator with Dirichlet boundary condition} to the interval $\Lambda_L(x):=(x-\ell L,x+\ell L)$, as the operator acting as $\DO $ on the domain
\begin{equation*}
\Dom\left(\DONxL \right):=\left\{\psi=\left(\begin{smallmatrix}
                          \psi^\uparrow\\\psi^\downarrow
                         \end{smallmatrix}\right)\in H^1\left(\Lambda_L(x),\R^{2N}\right)\text{ such that } \psi^\uparrow(x-\ell L)=\psi^\uparrow(x+\ell L)=0\right\}.\end{equation*}
                         We will use the notation $\DONL$ for $\DONOL$.

                         We define the density of states in the following way.
\begin{defi}
 For all compactly supported continuous functions $\phi$, the \emph{density of states} is 
 \[\nu(\phi)=\lim_{L\to\infty}\frac{1}{2\ell L}\ec(\tr(\phi(\DONL))).\]
\end{defi}
The well-definedness of this limit and its independence from $\omega$, consequence of the ergodic theorem, are proven in \cite[Appendix~B]{Z23} in the case $N=1$. The generalization to bigger $N$ is straightforward.

The function $\nu$ is a positive linear functional on the space of compactly supported continuous functions. By the Riesz-Markov theorem, it can be seen as an integral with respect to  some Borel measure on $\R$, which will be denoted by $\nu$ too.
The integrated density of states is defined on $\R$ by
\begin{equation}
 F(E):=\left\{\begin{array}{l}
        -\nu((E,0])\text{ if }E<0\\
        \nu((0,E])\text{ if } E\geq 0 
        \end{array}
 \right.
\end{equation}

We prove that, when the  sum of the Lyapunov exponents is H\"older continous, the integrated density of states is H\"older continuous as well:
\begin{thm}\label{thm:reguDOS}
Let $I$ be a compact interval and $\tilde{I}$ an open interval containing $I$. If the sum $\sum_{i=p}^N\gamma_p$ is H\"older continuous on $\tilde{I}$, then there exist $C'>0$ and $\alpha'>0$ such that
\begin{equation}\label{eq:DOS_Hold}
\forall E, E'\in I,\ |F(E)-F(E')|\leq C'|E-E'|^{\alpha'}.
\end{equation}
\end{thm}
The proof of Theorem~\ref{thm:reguDOS} is a combination of arguments of~\cite{Z23}, \cite{Bou08} and~\cite{SSB}. We give the main steps of it in Section~\ref{sec:reguDOS}.

The previous results will enable us to prove Anderson and dynamical localization for random Dirac operators, under the hypotheses of Theorem~\ref{thm:loc_criterion1} or~\ref{thm:loc_criterion2}. We use a method called \emph{bootstrap multiscale analysis}, which has been developed by Germinet and Klein in \cite{GKbootstrap} for a large class of operators. We present it in detail in Section~\ref{subsec:MSA}. In Sections~\ref{subsec:wegner} and~\ref{subsec:ILSE}, we explain how we get the conditions needed for the multiscale analysis from what we have proved before.


\subsection{Lyapunov exponents associated with subgroups of $\SpONN$}
We want to prove here that if a group $G$ is $2p$-contracting and $(J,S)$-$L_{2p}$-strongly irreducible for every $2p$ in $\{1,\dots,N\}$, then for all sequence of i.i.d. random matrices in $G$, 
\begin{enumerate}
    \item if $N$ is even, then $\mathbf{(L^{(N})}$ holds and the $N$-th Lyapunov exponent associated with the sequence is positive;
    \item if $N$ is odd, then $\mathbf{(L^{(N-1)})}$ holds and the $N$-th Lyapunov exponent associated with the sequence is zero.
\end{enumerate}
 We will use the properties of $\SpONN$ proven in Appendix~\ref{app:G5}. First, we prove that when $G$ is included in $\SpONN$,  all Lyapunov exponents are twice degenerate.
\begin{prop}\label{degen_gamma}
Let $(A^\omega_n)_{n\in\Z}$ be an i.i.d. sequence of matrices in $\SpONN$ and $(\gamma_1,\dots,\gamma_{2N})$ the associated family of Lyapunov exponents. Then, for all $p\leq N$, $\gamma_{2p}=\gamma_{2p-1}$. 
\end{prop}
\begin{proof}
We see from Proposition~\ref{prop:decompG}~(4) that the singular values of a matrix in $\SpONN$ always have multiplicity at least 2. This together with~\cite[Proposition~A.III.5.6]{BL} implies the degeneracy of the Lyapunov exponents.
\end{proof}

We denote, for even $N$, $N=:2d$ and, for odd $N$, $N=:2d+1$. 
The space $(J,S)$-$L_{p}$, introduced in Section~\ref{sec:loc-crit}, can be described in the following way.
 For $p\leq d$, we define \begin{equation}\label{def:f1}f_1:=\frac{1}{2^p}(e_1+e_2)\wedge\dots\wedge(e_{2p-1}+e_{2p})\wedge(e_{N+1}-e_{N+2})\wedge\dots\wedge(e_{N+2p-1}-e_{N+2p})\end{equation}
 and
 \begin{equation}\label{def:Lp}
           (J,S)\mbox{-}L_{2p}=\s\{\Lambda^{2p}Mf_1\, |\,M\in \SpONN\}.
           \end{equation}
          
We see in particular that $(J,S)$-$L_{2d}$ is stable under the action of $\Lambda^{2d}g$: $\mathbf{(L_1^{(2d)}})$ is satisfied. We prove $\mathbf{(L_2^{(2d)}})$--$\mathbf{(L_4^{(2d)}})$ with the following result.

\begin{prop}\label{prop:sepLyap}
 Let $(A_n^\omega)$ be a sequence of i.i.d. random matrices in $\SpONN$, of common law  $\mu$, and $p$ an integer between 1 and $d$. We denote by $G_\mu$ the Furstenberg group associated with the sequence $(A_n^\omega)$. We assume that it is $2p$-contracting and $(J,S)$-$L_{2p}$-strongly irreducible and that $\E(\log\|A_0^\omega\|)<\infty$. Then, 
 \begin{enumerate}
 \item for all $x\neq0$ in $L_{2p}$, \begin{equation}\label{conv_as}
     \lim_{n\to\infty}\frac1n\log\|\Lambda^{2p}(A_{n-1}^\omega\dots A_0^\omega)x\|=\sum_{i=1}^{2p}\gamma_i\text{ almost surely};
 \end{equation} 
 \item there exists a unique probability measure $\nu_{2p}$ on $\mathbb{P}(L_{2p})$ which is $\mu$-invariant,  such that
 \begin{equation}\label{inv_meas}
     \gamma_1+\dots+\gamma_{2p}=\int_{\SpONN\times \mathbb{P}(L_{2p})}\log\frac{\|(\Lambda^{2p}g)x\|}{\|x\|}\rd \mu(g)\rd \nu_{2p}(\bar{x}); \end{equation}
      \item \begin{equation}\label{sep_Lyap}
  \gamma_{2p}>\gamma_{2p+1}.
  \end{equation}
 \end{enumerate}
\end{prop}
Points (1) and (2) of Proposition~\ref{prop:sepLyap} for $p=d$ directly give $\mathbf{(L_2^{(2d)}})$ and $\mathbf{(L_3^{(2d)}})$. Moreover, Point (3) implies that $\gamma_{2d}>\gamma_{2d+1}$. If $N=2d$ is even, we have $\gamma_{N+1}=-\gamma_N$ and if $N=2d+1$ is odd, $\gamma_N=0$ by Proposition \ref{degen_gamma}. In both cases this implies that $\gamma_{2d}>0$ and \emph{a fortiori} $\sum_{p=1}^{2d}\gamma_p>0$ which is $\mathbf{(L_4^{(2d)}})$.

\begin{proof}[Proof of Proposition~\ref{prop:sepLyap}]
 Let $k$ be the dimension of $(J,S)$-$L_{2p}$. Given  an orthonormal basis $(f_1,\dots,f_k)$ of $(J,S)$-$L_{2p}$ ($f_1$ being defined by~\eqref{def:f1}), for each $M\in \SpONN$, we denote by $\widehat{M}$ the matrix of $\GL_k(\R)$ such that
 \begin{equation}
  \widehat{M}_{ij}=\langle f_i, \Lambda^{2p}Mf_j\rangle.
 \end{equation}

Let us now denote by $\widehat{G}_\mu$ the subgroup of $\GL_k(\R)$ which is generated by all $(\widehat{A}_0^\omega)$. We want to prove that $\widehat{G}_\mu$ satisfies the hypotheses of theorems given in Chapter~A.III of~\cite{BL}, which will give us information on the two top Lyapunov exponents associated with the sequence $(\widehat{A}_n^\omega)$. In a second time, we will establish a relation between the Lyapunov exponents associated with the sequence $(\widehat{A}_n^\omega)$  of matrices of $\GL_k(\R)$ and partial sums of the Lyapunov exponents associated with the sequence $(A_n^\omega)$ of matrices of $\SpONN$. 

Since $G_\mu$ is $(J,S)$-$L_{2p}$-strongly irreducible, $\widehat{G}_\mu$ is strongly irreducible, as a subset of $\GL_k(\R)$. Similarly, as $G_\mu$ is $2p$-contracting, $\widehat{G}_\mu$ is contracting. We can thus apply~\cite[Theorem~A.III.6.1]{BL} and find, denoting by $\hat{\gamma}_1$ and $\hat{\gamma}_2$ the two highest Lyapunov exponents associated with the sequence $(\widehat{A}_n^\omega)$, that \begin{equation}\label{sep_hat_gamma}\hat{\gamma}_1>\hat{\gamma}_2.\end{equation}
Moreover, Corollary~A.III.3.4 of~\cite{BL} gives that, for all $x\neq0$ in $\R^k$, \begin{equation}\label{hat_conv_as}
     \lim_{n\to\infty}\frac1n\log\|(\widehat{A}_{n-1}^\omega\dots \widehat{A}_0^\omega)x\|=\hat{\gamma}_1 \text{ almost surely}
 \end{equation} 
and, together with Theorem~A.III.4.3 of the same~\cite{BL}, that  there exists a unique probability measure $\hat{\nu}$ on $\mathbb{P}(\R^k)$ which is $\mu$-invariant,  such that
 \begin{equation}\label{hat_inv_meas}
     \hat{\gamma_1}=\int_{\SpONN\times \mathbb{P}(\R^k)}\log\frac{\|\hat{g}x\|}{\|x\|}\rd \mu(g)\rd \hat{\nu}(\bar{x}).
 \end{equation}

We now link the two Lyapunov exponents $\hat{\gamma}_1$ and $\hat{\gamma}_2$ with the Lyapunov exponents associated with the sequence $(A_n^\omega)$. In view of Definition~\eqref{eq_def_lyap_exp}, we have to estimate $\|\widehat{M}\|$ and $\|\Lambda^2\widehat{M}\|$ for matrices $\widehat{M}\in \widehat{G}_\mu$.

For such a matrix, we consider $t_1\geq\dots\geq t_p$  the values given in the decomposition of Proposition~\ref{prop:decompG}~(v). Let us prove that $\|\widehat{M}\|=\|\Lambda^{2p}M\|=\e^{2t_1}\dots \e^{2t_p}$: this will imply that $\hat{\gamma_1}=\sum_{i=1}^{2p}\gamma_i$. Since $(J,S)$-$L_{2p}$ is a subspace of $\Lambda^{2p}\R^{2N}$, we have that $\|\widehat{M}\|\leq\|\Lambda^{2p}M\|$.
 To prove the other inequality, recall that $f_1$, defined in~\eqref{def:f1}, is in $(J,S)$-$L_{2p}$.  As a consequence, we have, using the notation of Proposition~\ref{prop:decompG}~(5),
 \[\|\Lambda^{2p}M\|=\e^{2t_1}\dots \e^{2t_p}=\|\Lambda^{2p}Rf_1\|=\|\widehat{R}f_1\|\leq\|\widehat{R}\|\leq\|\Lambda^{2p}R\|=\|\Lambda^{2p}M\|.\] We used~\cite[Proposition~A.III.5.3]{BL} for the first equality and the fact that $(J,S)$-$L_{2p}$ is a subspace of $\Lambda^{2p}\R^{2N}$ for the last inequality. This implies that $\|\Lambda^{2p}M\|=\|\widehat{R}\|$. But, since the matrices $K$ and $U$ are orthogonal and in $\SpONN$, 
 $\|\widehat{R}\|=\|\widehat{K}\widehat{R}\widehat{U}\|=\|\widehat{M}\|$. As a consequence, \begin{equation}\label{est_gamma_1}\hat{\gamma}_1=\sum_{i=1}^{2p}  \gamma_i.\end{equation} 
 This last equality together with~\eqref{hat_conv_as} implies~\eqref{conv_as}. Together with~\eqref{hat_inv_meas}, \eqref{est_gamma_1} implies~\eqref{inv_meas}.
 
 In order to prove~\eqref{sep_Lyap}, we now estimate $\|\Lambda^2\widehat{M}\|$. We first notice that, since  $(J,S)$-$L_{2p}$ is a vector subspace of $\Lambda^{2p}\R^{2N}$, $\|\Lambda^2[\Lambda^{2p}M]\|\geq\|\Lambda^2\widehat{M}\|$.
 
 We begin with the case $p<d$. We choose
 \begin{align}
  f_2:=\frac{1}{2^p}&(e_1+e_2)\wedge\dots\wedge(e_{2p-3}+e_{2p-2})\wedge(e_{2p+1}+e_{2p+2})\\&\wedge(e_{N+1}-e_{N+2})\wedge\dots\wedge(e_{N+2p-3}-e_{N+2p-2})\wedge(e_{N+2p+1}-e_{N+2p+2}).\nonumber
 \end{align} 

 Let
\begin{equation}
M_0:= \left( \begin{matrix}
\begin{matrix}
I_{\mathrm{2p-2}}  & &  \\
 & \left( \begin{matrix}
 0 & I_2   \\
I_2 & 0  \end{matrix} \right) & \\
&  & I_{\mathrm{N-2p-2}}
\end{matrix} &\vline &  0 \\ \hline
0 &\vline &\begin{matrix}
I_{\mathrm{2p-2}}  & &  \\
 & \left( \begin{matrix}
 0 & I_2   \\
I_2 & 0  \end{matrix} \right) & \\
&  & I_{\mathrm{N-2p-2}}
\end{matrix}      
\end{matrix} \right). 
\end{equation}

Then $M_0\in \SpONN$ and $\wedge^{2p} M_0 f_1 = f_2$, hence $f_2\in (J,S)$-$L_{2p}$ and it is easy to see that $f_1$ and $f_2$ are orthogonal.
 Consequently, we can complete $(f_1,f_2)$ into an orthonormal basis $(f_1,\dots,f_k)$ of $\GL_k(\R)$ and use it in the definition of $\widehat{M}$.

We find that
\begin{equation}
 \|\Lambda^2\widehat{M}\|=\|\Lambda^2\widehat{R}\|\geq \|\Lambda^2\widehat{R}(f_1\wedge f_2)\|=\|\widehat{M}\|\e^{2t_1}\dots \e^{2t_{p-1}}\e^{2t_{p+1}}.
\end{equation}
In terms of Lyapunov exponents, this implies that $\hat{\gamma_1}+\hat{\gamma}_2\geq\hat{\gamma_1}+\sum_{i=1}^{2(p-1)}\gamma_i+2\gamma_{2p+1}$. This together with~\eqref{sep_hat_gamma} and \eqref{est_gamma_1} implies that $\gamma_{2p}>\gamma_{2p+1}$.

When $p=d$, we prove that $\gamma_{2d}>0$. This implies~\eqref{sep_Lyap} since, when $N$ is even, $\gamma_{2d+1}=\gamma_{N+1}\leq0$ and, when $N$ is odd, $\gamma_{2d+1}=\gamma_N=0$ as we have already seen.  We define 
\begin{align}
  f_2:=\frac{1}{2^d}&(e_1+e_2)\wedge\dots\wedge(e_{2d-3}+e_{2d-2})\wedge(e_{2d-1}-e_{2d})\\&\wedge(e_{N+1}-e_{N+2})\wedge\dots\wedge(e_{N+2d-3}-e_{N+2d-2})\wedge(e_{N+2d-1}+e_{N+2d}).\nonumber
 \end{align}
If we introduce 
 \begin{equation}
M_0':= \left( \begin{matrix}
\begin{matrix}
I_{\mathrm{N-1}}  &  \\
  & -1
\end{matrix} &\vline &  0 \\ \hline
0 &\vline &\begin{matrix}
I_{\mathrm{N-1}}  &  \\
  & -1
\end{matrix}      
\end{matrix} \right)\text{ (even $N$) or } 
M_0':= \left( \begin{matrix}
\begin{matrix}
I_{\mathrm{N-2}}  & &  \\
  & -1& \\ & &1
\end{matrix} &\vline &  0 \\ \hline
0 &\vline &\begin{matrix}
I_{\mathrm{N-2}}  & &  \\
  & -1& \\ & &1
\end{matrix}      
\end{matrix} \right)\text{ (odd $N$)}
\end{equation}
then $M_0'\in \SpONN$ and $f_2=\Lambda^{2d}f_1$ so $f_2\in (J,S)$-$L_{2p}$ with $f_2$ orthogonal to $f_1$.
As a consequence,
\begin{equation}
 \|\Lambda^2\widehat{M}\|=\|\Lambda^2\widehat{R}\|\geq \|\Lambda^2\widehat{R}(f_1\wedge f_2)\|=\|\widehat{M}\|e^{2t_1}\dots e^{2t_{d-1}}.
\end{equation}
In terms of Lyapunov exponents, this means that $\hat{\gamma}_2\geq \sum_{i=1}^{2(d-1)}\gamma_i$, which, together with \eqref{sep_hat_gamma} and \eqref{est_gamma_1}, implies that $\gamma_{2d}>0$.
\end{proof}

\subsection{Regularity of the Lyapunov exponents}\label{sec:reguLyap}
The goal of this section is to prove Theorem~\ref{thm:regu_Lyap}. It is a consequence of the following result. First, we note that the definitions of Lyapunov exponents and Furstenberg  group can be extended to the case of a sequence of i.i.d. random matrices in $\mathrm{GL}_{\mathrm{2N}}(\C)$. Hence one can state hypothesis $\mathbf{(L^{(N)})}$ not only for transfer matrices associated with an ergodic family of operators.

\begin{thm}\label{thm:reguLyap}
We fix a compact interval $I\subset \R$. Let $(A_n^\omega(E))_{n\in \Z}$ be a sequence of i.i.d. random  matrices in  $\SpC$ depending on a parameter $E$ in $I$. For each $E$, let $\mu_E$ be the common distribution of the $A_n^\omega(E)$.  We assume that, for all $E\in I$, $\ec(\log\|A_0^\omega(E)\|)<\infty$. We also assume that  Assumption~$\mathbf{(L^{(N)})}$ holds  on $I$, or that  Assumption~$\mathbf{(L^{(N-1)})}$ holds and $\gamma_N$ is identically $0$  on $I$. Moreover, we assume that,
\begin{enumerate}
 \item there exists  $C_1>0$ independent of $E$, $\omega$ and $n$ such that, for every $E\in I$,
 \begin{equation}\label{Gronwall1}
 \|\Lambda^NA_n^\omega(E)\|^2\leq C_1;
 \end{equation}
 \item there exists $C_2>0$ independent of $E$, $\omega$ and $n$ such that for every $E,E'\in I$
 \begin{equation}\label{Gronwall2}
 \|\Lambda^NA_n^\omega(E)-\Lambda^NA_n^\omega(E')\|^2\leq C_2|E-E'|.
 \end{equation}
 Then, there exist $C>0$ and $\alpha>0$ such that
 \begin{equation}
   \forall E, E'\in I,\quad|\gamma_1(E)+\dots+\gamma_N(E)-\gamma_1(E')-\dots-\gamma_N(E')|\leq C|E-E'|^\alpha.
 \end{equation}
\end{enumerate}

\end{thm}

This theorem is basically a  restatement of Theorem~1 of \cite{Bou08}.  A detailed proof can be found in~\cite[Section~6.3]{theseHB}. One difference is that we only prove H\"older regularity for the sum $\gamma_1+\dots+\gamma_N$, while the result of~\cite{Bou08} provides regularity for each of the $\gamma_p$'s. For this reason, it needs hypotheses for all $p$ while we need them only for $p=N$.
Then, our hypotheses (1) and (2) correspond to hypotheses (ii) and (iii) in~\cite{Bou08}. His hypothesis (i) is the fact that the Furstenberg group is $p$-contracting and $J$-$L_p$-strongly irreducible for all $p$ in $\{1,\dots,N\}$. If we carefully look at the proof of the theorem of~\cite{Bou08}, we see that  the only reason for which this is needed is because it implies  Assumption~$\mathbf{(L^{(N)})}$ on $I$. As a consequence, we can use this hypothesis instead of the one of $p$-contractivity and $J$-$L_p$-strong irreducibilty for all $p$, which makes it possible to apply the theorem in a more general setting. 

We want to apply Theorem \ref{thm:reguLyap} to the sequence $(A_n^\omega(E))_{n\in \Z}=(T_{\omega^{(n)}}(E))_{n\in \Z}$, the transfer matrices as defined in~\eqref{eq_def_transfer_mat}, to get Theorem \ref{thm:regu_Lyap}. Therefore, we have to prove that this sequence satisfies the hypotheses of Theorem~\ref{thm:reguLyap}.
We already assumed that  Assumption~$\mathbf{(L^{(N)})}$ holds  on $I$, or that  Assumption~$\mathbf{(L^{(N-1)})}$ holds and $\gamma_N$ is identically $0$  on $I$ in the hypotheses of Theorem~\ref{thm:regu_Lyap}. The boundedness of the potential implies that $\ec(\log\|T_{\omega^{(0)}}(E)\|)<\infty$ for all $E$.
 Finally, we are left with proving that the estimates \eqref{Gronwall1} and \eqref{Gronwall2} are always satisfied for the sequence $(T_{\omega^{(n)}}(E))_{n\in \Z}$. We begin with two lemmas which provides estimates on solutions of Dirac equations. Both have been proven in the case $N=1$ in \cite{Z23} and can easily be generalized to $N\geq 1$. As a consequence, we give only sketches of the proofs.
\begin{lemma}\label{lem:Gronwall1}
 Let $\psi=\begin{pmatrix}\psi^\uparrow\\\psi^\downarrow\end{pmatrix}$ be a solution to $\DON \psi+V\psi=0$, where $V$ is a $L_{loc}^1$ function with values in $2N$-by-$2N$ matrices. Then, for all $x$, $y\in\R$,
 \begin{equation} \label{majo1}
  |\psi^\uparrow(x)|^2+|\psi^\downarrow(x)|^2\leq  (|\psi^\uparrow(y)|^2+|\psi^\downarrow(y)|^2)\exp\left(2\int_{\min(x,y)}^{\max(x,y)}|V(t)|\dd t\right).
 \end{equation}
\end{lemma}
\begin{proof}
This is  a straightforward adaptation of~\cite[Lemma~4.3]{Z23}. The key argument is Gr\"onwall's lemma applied  to the function $|\psi^\uparrow(x)|^2+|\psi^\downarrow(x)|^2$.\end{proof}
 

\begin{lemma}\label{lem:Gronwall2}
Let $\psi_1=\begin{pmatrix}\psi_1^\uparrow\\\psi_1^\downarrow\end{pmatrix}$ and $\psi_2=\begin{pmatrix}\psi_2^\uparrow\\\psi_2^\downarrow\end{pmatrix}$ be solutions to $\DON \psi_i+V_i\psi_i=0$, where, for $i=1,2$, $V_i$ is a $L_{loc}^1$ function with values in $2N$-by-$2N$ matrices, and such that, for some $y$ in~$\R$, $\psi_1(y)=\psi_2(y)$. Then, we have, for all $x$ in $\R$,
\begin{equation}
 |\psi_1(x)-\psi_2(x)|\leq|\psi_2(y)| \exp\left(\int_{\min(y,x)}^{\max(y,x)}|V_1(t)|+|V_2(t)|\d t\right)\times\int_{\min(y,x)}^{\max(y,x)}|V_1(s)-V_2(s)|\d s.
\end{equation}

\end{lemma}
\begin{proof}
The proof is similar to the one of Lemma~4.4 in~\cite{Z23}. It relies once again on Gr\"onwall's lemma, applied here to the function $\psi_1-\psi_2$.
\end{proof}

We get~\eqref{Gronwall1} and~\eqref{Gronwall2} applying the arguments  of~\cite[Section~6.2.2.1]{theseHB}. We prove that $\|A_n^\omega(E)\|^2\leq C_1$ and that $\|A_n^\omega(E)-A_n^\omega(E')\|^2\leq C_2|E-E'|$ by applying respectively Lemmas~\ref{lem:Gronwall1} and~\ref{lem:Gronwall2} to each column of the transfer matrix, which is a solution of $\DON \psi+(V_\omega-E)\psi=0$ (\emph{cf.} Lemmas~6.2.3 and~6.2.5 of~\cite{theseHB}). We get the estimates with the external power as in Lemma~6.2.6 of~\cite{theseHB}: \eqref{Gronwall1} directly comes from the fact that, for all invertible matrix $M$, $\|\Lambda^NM\|\leq\|M\|^N$, while the proof of~\eqref{Gronwall2} uses that, for all invertible matrices $M_1$ and $M_2$,
\[\|\Lambda^NM_1-\Lambda^NM_2\|\leq\|M_1-M_2\|\left(\|M_1\|^{N-1}+\|M_1\|^{N-2}\|M_2\|+\dots+\|M_2\|^{N-1}\right).\]
This concludes the proof of \eqref{Gronwall1} and \eqref{Gronwall2} and thus of Theorem~\ref{thm:regu_Lyap}.

\subsection{H\"older continuity of the integrated density of states}\label{sec:reguDOS}
We prove in this section Theorem~\ref{thm:reguDOS}. The link between the Lyapunov exponents and the integrated density of states is given by the following Thouless formula, similar to~\cite[Proposition~5.2]{Z23} and \cite[Theorem~3]{Bou08}.
\begin{prop}\label{prop:thouless}
 Let $F$ and $F_0$ be the integrated densities of states, respectively of  $\{\DO \}_{\omega \in \Omega}$ and $\DON$, and $\gamma_1,\dots,\gamma_N$ be the Lyapunov exponents of $\{\DO \}_{\omega \in \Omega}$. Then, there exists $a\in\R$ such that, for every almost-every $E\in\R$,
 \begin{equation}\label{formule:thouless}
 \sum_{i=1}^N\gamma_i(E)=-a+\int_\R\log\left|\frac{E-t}{t-\mathrm{i}}\right|\d(F-F_0)(t).
 \end{equation}
Moreover, if $I\subset \R$ is an interval on which $E\mapsto (\gamma_{1}+\ldots +\gamma_{N})(E)$ is continuous then \eqref{formule:thouless} holds for every $E\in I$.
\end{prop}
Note that, since the spectrum of $\DON$ is purely absolutely continuous  with multiplicity $2N$ on the whole real line, all its Lyapunov exponents are identically zero for real energies (\textit{cf.}~\cite[Corollary~VII.3.4]{CL90}). That is why they do not appear in the formulas.
 The integrated density of states $F_0$ of the free Dirac operator $\DON$ can be  explicitly computed. Indeeed,  the eigenvalues of $\DON$ restricted to $[-\ell L,\ell L]$ with Dirichlet boundary conditions are the $(k\ell L)_{k\in\Z}$, each of them with multiplicity $N$. As a consequence, $F_0(E)=NE/\pi$.

The following lemma, which has been proven in  the case $N=1$ in~\cite{Z23}, holds in fact for all $N$. It makes it possible to control the difference $F(E)-F_0(E)$.
\begin{lemma}[\cite{Z23}, Lemma~5.3]
 There exists a constant $C$, depending only on $\|V_\omega\|_\infty$, such that for all $E\in\R$, we have
 \begin{equation}
  |F(E)-F_0(E)|\leq C.
 \end{equation}
 \end{lemma}

The next step towards the Thouless formula is to introduce a tool named Kotani $w$-function. We use it in the way it has been introduced by Sadel and Schulz-Baldes for Dirac operators in \cite[Section~5]{SSB}. The first step is to introduced the so-called Weyl-Titchmarsh matrix. Existence and uniqueness of this matrix are proven in~\cite[Theorem~2]{SSB}.
\begin{defi}
Let $z\in\C\backslash\R$. Then,
there exists a unique $M_\omega(z)\in\mathcal{M}_N(\C)$, called the \emph{Weyl-Titchmarsh matrix} such that
the space of solutions of $\DO \psi=z\psi$ which are square-integrable on $(0,+\infty)$  is spanned by the solutions whose value in 0 is one of the column of $\left(\begin{smallmatrix}
I_N\\M_\omega(z)\end{smallmatrix}\right)$.
\end{defi}

In order to be able to use the results of Sadel and Schulz-Baldes, we need to have an $\R$-ergodic operator, which is \emph{a priori} stronger than being $\Z$-ergodic. Nevertheless, in \cite{Kirsch}, Kirsch proves that, given a $\Z$-ergodic operator $H_\omega$, we can construct an $\R$-ergodic operator $\overline{H}_{\overline{\omega}}$, on a wider probability space, with, for each $\overline{\omega}$,  $\overline{H}_{\overline{\omega}}$ is unitarily equivalent to $H_\omega$ for some $\omega$.
We apply in our setting Kirsch's suspension procedure. Therefore, we introduce a new expectation $\eb$ which  must here be understood as both expectation on $\Omega$ and average value on $[0,\ell]$ for the potentials. 

As in~\cite{SSB}, we introduce a block decomposition for the potential: 
\[V_\omega=:\begin{pmatrix}P_\omega&R_\omega\\R_\omega^*&Q_\omega
\end{pmatrix}.\]
The $w$ function is defined on $\C\setminus\R$ (with values in $\C$) as follows:
\begin{equation}
   w(z):=-\eb\left[\tr(R_\omega+M_\omega(z)(Q_\omega-z))\right].
\end{equation}
This function will provide us a link between the Lyapunov exponent and the integrated density of states through the Green's function. We recall that the Green's function $G(z,\cdot,\cdot)$ is the integral kernel of the resolvent $(\DO-z) ^{-1}$.  We start by  the following theorem of Sadel and Schulz-Baldes, which links the Green's function and the Lyapunov exponent. Here $G(z):=G(z,x,x)$, and we drop the dependence in $x$ since we take the expectation $\eb$.
\begin{thm}[\cite{SSB}, Theorem~5]\label{thm:SSB}
 Let $\Im(z)\neq0$. Then,
 \begin{enumerate}
  \item $\sum_{i=1}^N\gamma_i(z)=-\Re(w(z))$.
  \item $\partial_zw(z)=\eb(\tr~ G(z))$.
 \end{enumerate}
\end{thm}

Once we have this, proving Proposition~\ref{prop:thouless} can be done as in~\cite[Section~5]{Z23} (which considers the case $N=1$) and~\cite[Section~4.2]{Bou08} (which considers Schr\"odinger operators for general $N$). As a consequence, we give here only the main steps of the proof and let the interested reader refer to the cited papers for details. We first establish a relation between the density of states measure $\nu$ and the Green's function.
\begin{lemma}[\cite{Z23}, Lemma~5.5,\cite{Bou08}, Proposition~9]\label{lem:G-IDS}
 Let $z\in\C\backslash\R$. Then,
 \[\eb(\tr(G(z)))=\int_\R\frac{\d\nu(E)}{E-z}.\]
\end{lemma}

This together with the second point of Theorem~\ref{thm:SSB} give a relation between the imaginary part of the $w$ function and the integrated density of states.

\begin{lemma}[\cite{Z23}, Lemma~5.6, \cite{Bou08}, Proposition~10]
 There exists $c\in\R$ such that, for all $E\in\R$, we have
 \begin{equation}
  \lim_{a\to0^+} \Im\left(w(E+\mathrm{i} a)-w_0(E+\mathrm{i} a)\right)=\pi(F(E)-F_0(E))+c.
 \end{equation}
\end{lemma}

We have an analogous result for the real part of $w$.
\begin{lemma}\label{rew}
For Lebesgue-almost every $E$ in $\R$, we have : 
\begin{equation}\label{reweq}
\lim_{a \to 0^{+}} \mathrm{Re}\ (w(E+\mathrm{i}a)-w_0(E+\mathrm{i} a))=-(\gamma_{1}+\ldots +\gamma_{N})(E).
\end{equation}
Moreover, if $I\subset \R$ is an interval on which $E\mapsto -(\gamma_{1}+\ldots +\gamma_{N})(E)$ is continuous then (\ref{reweq}) holds for every $E\in I$.
\end{lemma}

We can thus conclude the proof of Proposition~\ref{prop:thouless} with the arguments of~\cite{Bou08} and~\cite{Z23}. We then use the Hilbert transformation as in~\cite{Bou08} and~\cite{Z23} to get, through the Thouless formula proven in Proposition~\ref{prop:thouless}, the H\"older regularity of the integrated density of states from the one of the sum of the $N$ positive Lyapunov exponents. This achieve the proof of Theorem \ref{thm:reguDOS}.

\subsection{Multiscale analysis and localization}\label{sec:MSA}
In this last subsection, we explain how we can apply the multiscale analysis to get Anderson and dynamical localization in our setting.
\subsubsection{Conditions for the multiscale analysis}\label{subsec:MSA} 
We state here the hypotheses and the theorems on the multiscale analysis.
 We follow the presentation of the lecture notes~\cite{Kl08}. Then, we explain how we can apply them to our case. In the following paragraphs, $\{H_\omega \}_{\omega \in \Omega}$ will be any self-adjoint random ergodic operator on $L^2(\R, \C^{2N})$.

The first property gives the possibility to use generalized eigenfunctions.

Let $\Hilb:=L^2(\R,\d x;\C^{2N})$. Given $\nu>1/4$, we put, for $x\in\R$, $\langle x\rangle:=\sqrt{1+x^2}$ and we define the weighted spaces $\Hilb_\pm$ as
 $$
 \Hilb_\pm:=L^2(\R,\langle x\rangle ^{\pm4\nu}\d x;\C^{2N}).
 $$

 We define a duality map between $\Hilb_+$ and $\Hilb_-$ by the following sesquilinear form, where $\phi_1\in\Hilb_+$ and $\phi_2\in\Hilb_-$:
 $$
 \langle\phi_1,\phi_2\rangle_{\Hilb_+,\Hilb_-}:=\int\bar{\phi_1}(x)\cdot\phi_2(x)\d x.
 $$

We set $T$ to be the self-adjoint operator on $\Hilb$ given by multiplication by the function $\langle x\rangle ^{2\nu}$; note that $T^{-1}$ is bounded.

\begin{property}[SGEE]
 We say that an ergodic random operator $\{H_\omega \}_{\omega \in \Omega}$ satisfies the strong property of  generalized eigenfunction expansion (SGEE) in some open interval $I$ if, for some $\nu>1/4$, 
 \begin{enumerate}
  \item The set
 $$
 \Dom^\omega_+:=\{\phi\in\Dom(H_\omega)\cap \mathcal{H}_+; 
 H_\omega\phi\in \Hilb_+\}
 $$
 is dense in $\Hilb_+$ and is an operator core for $H_\omega$ with probability one.
 
 \item There exists a bounded, continuous function $f$ on $\R$, strictly positive on the spectrum of $H_\omega$,
 such that
 \begin{equation*}
  \ec\left\{[ \tr(T^{-1}f(H_\omega) \Pi_\omega(I)T^{-1})]^2\right\}<\infty,
 \end{equation*}
$\Pi_\omega$ being the spectral projection associated with $H_\omega$.
 \end{enumerate}
\end{property}
%
%
%
\begin{defi}
A measurable function $\psi:\R\to\C^{2N}$ is said to be a \emph{generalized eigenfunction} of $H_\omega$ with \emph{generalized eigenvalue} $E$ if $\psi\in\Hilb_-\backslash\{0\}$ 
and 
 $$
 \langle H_\omega\phi,\psi \rangle_{\Hilb_+,\Hilb_-}=E\langle \phi,
 \psi \rangle_{\Hilb_+,\Hilb_-},\ \mbox{ for all } \phi\in\Dom_+^\omega.
 $$
\end{defi}
%
%
%
As explained in \cite{Kl08}, when (SGEE) holds, a generalized eigenfunction which is in $\Hilb$ is a \emph{bona fide} eigenfunction. 
Moreover, if $\mu_\omega$ is the spectral measure for the restriction of $H_\omega$ to the Hilbert space $E_\omega(I)\Hilb$, then $\mu_\omega$-almost every $E$ is a generalized eigenvalue of $H_\omega$.

The following properties are about finite volume operators, restricted to
  intervals of the form $\Lambda_L(x):=[x-\ell L,x+\ell L]$.

\begin{defi}
We say that  an ergodic random family of operators $\{H_\omega\}_{\omega\in \Omega}$ is \emph{standard}  if for each $x\in\Z$, $L\in\N$ there is a measurable map $H_{\cdot,x,L}$ from $\Omega$ to self-adjoint operators
 on $L^2(\Lambda_L(x),\C^{2N})$ such that 
 $$
   U(y)H_{\omega,x,L}U(-y)=H_{\tau_y(\omega),x+y,L}
 $$ 
 where $\tau$ and $U$ define the ergodicity:
 $$
 U(y)H_\omega U(y)^*=H_{\tau_y(\omega)}.
 $$
\end{defi}
We can then define $R_{\omega,x,L}(z):=(H_{\omega,x,L}-z)^{-1}$ as the resolvent of 
$H_{\omega,x,L}$ and $\Pi_{\omega,x,L}(\cdot)$ as its spectral projection.

We now enumerate the properties which are needed for multiscale analysis to be performed, yielding thus various localization properties.

\begin{defi}
An event is said to be \emph{based in a box} $\Lambda_L(x)$ if it is determined by 
conditions on the finite volume operators 
$(H_{\omega,x,L})_{\omega\in\Omega}$.
\end{defi}
\begin{property}[IAD]
Events based in disjoint boxes are independent.
\end{property}

The following properties are to hold in a fixed open interval $I$. We will denote by $\chi_{x,L}$ the characteristic function of $\Lambda_L(x)$ and 
$\chi_x := \chi_{x,1}$. We also denote  by $\Gamma_{x,L}$ the characteristic function of the union of two regions near the boundary of $\Lambda_L(x)$: $[x-\ell(L-1),x-\ell(L-3)]\cup[x+\ell(L-3),x+\ell(L-1)]$. 
\begin{property}[SLI]\label{SLI}
For any compact interval $J\subset I$ there exists a finite constant $\kappa_J$ such that, given $L$, $l'$, $l''\in 2\N$, $x$, $y$, $y'\in\Z$ with 
 $\Lambda_{l''}(y)\subset\Lambda_{l'-3}(y')\subset\Lambda_{L-6}(x)$, then, for $\mathbb{P}$-almost every $\omega$, if $E\in J$ with $E\notin\sigma(H_{\omega,x,L })\cup\sigma(H_{\omega,y',l'})$
 we have
 \begin{equation}\label{eq:SLI-cond}
  \|\Gamma_{x,L}R_{\omega,x,L}(E)\chi_{y,l''}\|\leqslant\kappa_J\|\Gamma_{y',l'}R_{\omega,y',l'}(E)\chi_{y,l''}\|\|\Gamma_{x,L}R_{\omega,x,L}(E)\Gamma_{y',l'}\|.
 \end{equation}

\end{property}

\begin{property}[EDI]
 For any compact interval $J\subset I$ there exists a finite constant $\tilde{\kappa}_J$ such that for $\mathbb{P}$-almost every $\omega$, given a generalized eigenfunction $\psi$ of $H_\omega$
 with generalized eigenvalue $E\in J$, we have, for any $x\in\Z$ and $L\in2\N$ with $E\notin\sigma(H_{\omega,x,L})$, that
 \begin{equation*}
  \|\chi_x\psi\|\leqslant\tilde{\kappa}_J\|\Gamma_{x,L}
  R_{\omega,x,L}(E)\chi_x\|\|\Gamma_{x,L}\psi\|.
 \end{equation*}
\end{property}

\begin{property}[NE]
 For any compact interval $J\subset I$ there exists a finite constant $C_J$ such that, for all $x\in\Z$ and $L\in2\N$,
 \begin{equation*}
  \ec \left(\tr \left(\Pi_{\omega,x,L}(J)\right)\right)\leqslant C_J L.
 \end{equation*}

\end{property}

%
\begin{property}[W]\label{propWeg}
 For every $\beta\in(0,1)$ and every $\sigma>0$, there exists $L_0\in\N$ and $\alpha>0$ such that
 \begin{equation}\label{eq_propWeg}
  \pc(\dist(E,\sigma(H_{\omega,x,L}))\leq \e^{-\sigma L^\beta})\leq\e^{-\alpha L^\beta}
 \end{equation}
for all $E\in I$, $x\in\Z$ and $L\geq L_0$.
\end{property}
This Wegner estimate is not exactly similar to the one stated in \cite{Kl08}, but it is possible to use this version for multiscale analysis, as explained in~\cite{CKM} and~\cite[Remark~4.6]{Kl08}.

The last property depends on several parameters: $\theta>1$, $E_0\in\R$ and $L_0\in6\N$.
\begin{property}[ILSE($\theta$, $E_0$, $L_0$)]
 \begin{equation}\label{eq:ILSE}
  \pc\left\{\left\|\Gamma_{0,L_0}R_{\omega,0,L_0}(E_0)
  \chi_{0,L_0/3}\right\|
  \leqslant\frac{1}{L_0^\theta}\right\}>1-\frac{1}{841}.
 \end{equation}
\end{property}

 These properties are the hypotheses of the bootstrap multiscale analysis.
\begin{defi}
Given $E\in\R$, $x\in\Z$ and $L\in6\N$ with $E\notin \sigma(H_{\omega,x,L})$, we say that the box $\Lambda_L(x)$ is $(\omega,m,E)$-regular for a given $m>0$ if
\begin{equation}\label{eq:regular}
  \left\|\Gamma_{x,L}R_{\omega,x,L}(E)\chi_{x,L/3} \right\|
  \leqslant \e^{-mL}.
\end{equation}
\end{defi}

%
%
%
In the following, we denote
 $$
 [L]_{3\N}=\sup\{n\in3\N | n\leqslant L\}.
 $$
\begin{defi}
For $x$, $y\in\Z$, $L\in3\N$, $m>0$ and $I\subset\R$ an interval, we denote
\begin{equation*}
 R(m,L,I,x,y) 
  =\left\{\omega;\text{for every }E'\in I\text{ either }\Lambda_L(x)\text{ or }\Lambda_L(y)\text{ is }(\omega,m,E')\text{-regular.}\right\}.
\end{equation*}
The multiscale analysis region $\Sigma_{MSA}$ for $\{H_\omega \}_{\omega \in \Omega}$ is the set of $E\in\Sigma$  for which there exists some open interval $I \ni E$
such that, given any $\zeta$, $0<\zeta<1$ and $\alpha$, $1<\alpha<\zeta^{-1}$, there is a length scale $L_0\in3\N$ and a mass $m>0$ so if we set $L_{k+1}=[L_k^\alpha]_{3\N}$, $k=0,1, \ldots $,
we have
\begin{equation*}
 \mathbb{P}\left\{R(m,L_k,I,x,y)\right\}\geqslant 1-e^{-L_k^\zeta}
\end{equation*}
for all $k\in\N$, $x$, $y\in\Z^d$ with $|x-y|>L_k$.
\end{defi}
%
%
%

\begin{thm}[Multiscale analysis - Theorem 5.4 p136 of \cite{Kl08}]\label{pi}
 Let $\{H_\omega \}_{\omega \in \Omega}$ be a standard ergodic random operator with (IAD) and properties (SLI), (NE) and (W) fulfilled in an open interval $I$ and  let $\Sigma$ be the almost sure spectrum of $\{H_\omega \}_{\omega \in \Omega}$.  Given $\theta>1$, for each $E\in I$ there exists a finite scale
 $\mathcal{L}_\theta(E)=\mathcal{L}_\theta(E)>0$ bounded on compact subintervals of $I$ such that, if for a  given $E_0\in\Sigma\cap I$ we have (ILSE)($\theta$, $E_0$, $L_0$) at some scale $L_0\in3\N$ with $L_0>\mathcal{L}_\theta(E_0)$,
 then $E_0\in\Sigma_{MSA}$.
\end{thm}
%
%
%

\begin{thm}[Localization - Theorem~6.1 p139 of \cite{Kl08}]\label{thm:main-section4}
 Let $\{H_\omega \}_{\omega \in \Omega}$ be a standard ergodic operator with (IAD) and properties (SGEE) and (EDI) in an open interval $I$. Then, $\{H_\omega \}_{\omega \in \Omega}$ exhibits both dynamical and Anderson localization on $
 \Sigma_{MSA}\cap I$.
\end{thm}
%
%

We can summarize in the following figure, the ingredients of a proof of localization using multiscale analysis.

\begin{equation}\label{fig_localization_plan}
\underbrace{\mathrm{(IAD)}+\mathrm{(SLI)}+\mathrm{(NE)}+\mathrm{(W)}+\mathrm{(ILSE)}}_{\Downarrow}
\end{equation} 
$$\qquad \qquad \qquad \qquad \ \underbrace{\mathrm{(MSA)}+\mathrm{(SGEE)}+\mathrm{(EDI)}}_{\Downarrow} $$
$$\qquad \qquad \qquad \qquad \ \overbrace{\mathrm{(AL)}+\mathrm{(DL)}}$$
\bigskip

These theorems make it possible to prove Theorem~\ref{thm:loc_criterion1} and Theorem~\ref{thm:loc_criterion2}, by applying the multiscale analysis to the random family of operators $\{\DO \}_{\omega \in \Omega}$ defined by~\eqref{def_op_D_general}.


\subsubsection{The Wegner estimate}\label{subsec:wegner}
We begin with the Wegner estimate: we will get it from the regularity of the integrated density of states. The proof is a straightforward adaptation of what we have already done in~\cite{Bou09} and~\cite{Z23}. As a consequence, we only give a sketch of it. 

We prove the following proposition.
\begin{prop}\label{prop:wegner}
    Let $I$ be a compact interval included in an open interval $\tilde I$ on which Assumption~$\mathbf{(L^{(N)})}$ holds or on which Assumption~$\mathbf{(L^{(N-1)})}$ holds and $\gamma_N$ is identically $0$. Then (W) is satisfied on $I$.
\end{prop}

If $I$ is such an interval, we know from the previous sections  that \eqref{eq:DOS_Hold} is satisfied on $I$. This implies the following lemma. It is proven in~\cite{Z23} in the case $N=1$, but its generalization to higher $N$ is straightforward.
\begin{lemma}[\cite{Z23}, Lemma~6.16]\label{lemWeg}
 There exists $\rho>0$ and $C<\infty$ such that, for every $E\in I$ and every $\epsilon>0$, we have for $L\geq L_0$
 \begin{equation}\label{proba1}\begin{split}
 & \pc\{\text{There exists }E'\in(E-\epsilon,E+\epsilon)\text{ and }\phi\in\Dom\left(\DONL \right),\|\phi\|=1,\text{ such that }\\& (\DONL-E')\phi=0, |\phi^\downarrow(-\ell L)|^2+|\phi^\downarrow(\ell L)|^2\leq\epsilon^2\}\leq CL\epsilon^\rho.\end{split}
 \end{equation}
\end{lemma}
Moreover, hypotheses of Proposition \ref{prop:wegner} also implies Lemma~3 of~\cite{Bou09}, which is \emph{a priori} only applicable random \emph{Schr\"odinger} operators, but can also be applied to our setting since its proof only consists of algebraic manipulations of transfer matrices.
\begin{lemma}[\cite{Bou09}, Lemma~3]
    For all $p=1,\dots,N$, there exist $\xi_1>0$, $\delta>0$ and $n_1\in\N$ such that, for every $E\in I$, $n\geq n_1$ and $x\in \Lambda^p\C^{2N}$ with $\|x\|=1$, we have 
    \begin{equation}
        \ec\left(\|\Lambda^p\Phi_E(n,\omega)x\|^{-\delta}\right)\leq\e^{-\xi_1n}.
    \end{equation}
\end{lemma}

From these two lemmas, it is possible to prove  Proposition~\ref{prop:wegner}  following line by line the arguments of~\cite[Section~5]{Bou09} or~\cite[Section~6.2]{Z23}.

\subsubsection{The Initial Length Scale Estimate}\label{subsec:ILSE}

The last ingredient in order to apply a multiscale analysis scheme to our models is an ILSE which allows to start the induction proof. We denote by $\Lambda_L$ the interval $\Lambda_L(0)$.
We will prove an estimate stronger that~\eqref{eq:ILSE}, namely that there exists $L_0\in 3\N$ and $C,c, \delta>0$ such that, for all $L\geq L_0$ and all $E\in I$,
\begin{equation}\label{eq:ILSEstrong}
     \pc\left\{\left\|\Gamma_{0,L}R_{\omega,0,L}(E_0)
  \chi_{0,L/3}\right\|
  \leqslant \e^{-\delta L}\right\}\geq 1-C\e^{-cL}.
\end{equation}

As in most previous works, we will use the fact that the resolvent of the restricted operator has an integral kernel. 
We prove as in Equation~(6.11) of~\cite{Z23} that the resolvent $R_{\omega,0,L}(E)$ has an integral kernel, \emph{i.e.} there exists a function $G_{\Lambda_L}^\omega$ such that for all $\psi$ in $L^2(\Lambda_L)$ and almost every $y\in \Lambda_L$
\begin{equation}\label{def:int_ker}(\DONL-E)\int_{-\ell L}^{\ell L}G_{\Lambda_{L}}^\omega(E,x,y)\psi(x)\d x=\psi(y).\end{equation}
Note that~\cite{Z23} provides an explicit formula for this integral kernel, but we will not use it here.
We know from the Schur test that
\begin{equation}\label{eq:schur}
   \left\|\Gamma_{0,L}R_{\omega,0,L}(E)
  \chi_{0,L/3}\right\|\leq L \sup_{x,y\in\Lambda_L}|\Gamma_{0,L}(y)G_{\Lambda_{L}}^\omega(E,x,y)
  \chi_{0,L/3}(x)|.
\end{equation}
Our goal is therefore to bound $|G_{\Lambda_{L}}^\omega(E,x,y)|$ for $x$ and $y$ respectively in the support of $\chi_{0,L/3}$ and $\Gamma_{0,L}$, with probability exponentially close to $1$.

To this purpose, we give an explicit formula for $G_{\Lambda_{L}}^\omega$ at these points. For $E\in I$, we consider the two functions $\Phi_\pm=\left(\begin{smallmatrix}
    \Phi_\pm^\uparrow\\\Phi_\pm^\downarrow
\end{smallmatrix}\right)$ with values in $\mathcal{M}_{2\mathrm{N},\mathrm{N}}(\C)$, satisfying $\DO \Phi_\pm=E\Phi_\pm$ with 
\begin{equation}\label{bc_Phi_ILSE}
 \Phi_-(-\ell L)=\left(\begin{smallmatrix}
  0\\ I_{\mathrm{N}} 
\end{smallmatrix}\right)   \ \mbox{ and }\ \Phi_+(\ell L)=\left(\begin{smallmatrix}
    0\\ I_{\mathrm{N}} 
\end{smallmatrix}\right).  
\end{equation}


\begin{lemma}\label{lem:Green_kerD}
Let $\omega \in \Omega$ and let $x,y \in \Lambda_L$. Assume that $\Phi_+^\uparrow(x)$, $\Phi_+^\downarrow(x)$,  $\Phi_-^\uparrow(x)$, and $\Phi_-^\downarrow(x)$  are invertible as well as $\Phi_+^\downarrow(x)\Phi_+^\uparrow(x)^{-1}-\Phi_-^\downarrow(x)\Phi_-^\uparrow(x)^{-1}$ and $\Phi_+^\uparrow(x)\Phi_+^\downarrow(x)^{-1}-\Phi_-^\uparrow(x)\Phi_-^\downarrow(x)^{-1}$. The Green kernel of $\DONL$ is given by
 \begin{equation}\label{defGD}
  G_{\Lambda_{L}}^\omega(E,x,y)=\left\{\begin{array}{r l }\Phi_+(y)\alpha_+(x)\ &\text{ for }x\leq y\\
                            \Phi_-(y)\alpha_-(x)\ &\text{ for }x>y
                            \end{array}\right.
 \end{equation}
 where
 \begin{equation}\label{def:alpha+}
     \alpha_+(x):=\left(\begin{smallmatrix}
    \Phi_+^\uparrow(x)^{-1}\left(\Phi_+^\downarrow(x)\Phi_+^\uparrow(x)^{-1}-\Phi_-^\downarrow(x)\Phi_-^\uparrow(x)^{-1}\right)^{-1}, \Phi_+^\downarrow(x)^{-1}\left(\Phi_-^\uparrow(x)\Phi_-^\downarrow(x)^{-1}-\Phi_+^\uparrow(x)\Phi_+^\downarrow(x)^{-1}\right)^{-1} \end{smallmatrix}\right)
 \end{equation}
 and
  \begin{equation}\label{def:alpha-}
     \alpha_-(x):=\left(\begin{smallmatrix}\Phi_-^\uparrow(x)^{-1}\left(\Phi_+^\downarrow(x)\Phi_+^\uparrow(x)^{-1}-\Phi_-^\downarrow(x)\Phi_-^\uparrow(x)^{-1}\right)^{-1}, \Phi_-^\downarrow(x)^{-1}\left(\Phi_-^\uparrow(x)\Phi_-^\downarrow(x)^{-1}-\Phi_+^\uparrow(x)\Phi_+^\downarrow(x)^{-1}\right)^{-1}\end{smallmatrix}\right)
 \end{equation}
 \end{lemma}
 \begin{proof}
 By definition, \eqref{def:int_ker} has to be satisfied for all $\psi\in L^2(\Lambda_L)$ and almost all $x\in\Lambda_L$.
We compute this expression explicitly, assuming that $ G_{\Lambda_{L}}^\omega$ has the form given by \eqref{defGD} with $\alpha_\pm=(\alpha_\pm^\leftarrow,\alpha_\pm^\rightarrow)$. Using the fact that $(\DONL-E)\Phi_\pm=0$, we find that for all $x\in\Lambda_L$
\begin{equation}
  \begin{pmatrix}
      \phi_+^\uparrow(x)\alpha_+^\leftarrow(x)-\phi_-^\uparrow(x)\alpha_-^\leftarrow(x)& \phi_+^\uparrow(x)\alpha_+^\rightarrow(x)-\phi_-^\uparrow(x)\alpha_-^\rightarrow(x)\\ \phi_+^\downarrow(x)\alpha_+^\leftarrow(x)-\phi_-^\downarrow(x)\alpha_-^\leftarrow(x)& \phi_+^\downarrow(x)\alpha_+^\rightarrow(x)-\phi_-^\downarrow(x)\alpha_-^\rightarrow(x)
  \end{pmatrix}=\begin{pmatrix}
      0&-I_N\\I_N&0
  \end{pmatrix}
  \end{equation}
Assuming the invertibility conditions given in the statement of the lemma, we can solve this linear system explicitly, which gives the expressions~\eqref{def:alpha+} and~\eqref{def:alpha-}.
 \end{proof}


We now introduce two events on which we will be able to bound the integral kernel. Their probability will be controlled by either Proposition \ref{prop_LDP_valeurssingulieres_SpC} or Proposition \ref{prop_LDP_valeurssingulieres_SpONN}. In order to estimate blocks of the products of transfer matrices as in \cite{MaSo22}, we  introduce, given a vector subspace $F$ of $\C^{2N}$, the orthogonal projection onto $F$  $\pi_F : \C^{2N} \to F$  and we set
$$\pi_F^*\ :\  \begin{array}{ccl}
   F  & \to & \C^{2N} \\
   x & \mapsto & x
\end{array}.$$

For any integer $n\in [-L/3, L/3]$, any $T\in \mathcal{M}_{2\mathrm{N}}(\C)$ and any vector subspace $F\subset \C^{2N}$, we define
\begin{align}
    \Omega_\epsilon^F[T]:=& \left\{\max_{1\leq p\leq N}\left(  \left|\frac{1}{\ell|n-L|}\log s_p\left(T\right)-\gamma_p(E)\right| + \right. \right. \\ \nonumber 
    & \qquad \qquad \qquad \qquad \left. \left. \left|\frac{1}{\ell|n-L|}\log s_p\left(T\pi_F^*\right)-\gamma_p(E)\right| \right)\leq \frac{\epsilon}{100N}\right\}.\label{eq_def_Omega_F_epsilon}
\end{align}
where $s_p(\cdot)$ denotes the $p$-th singular value of the considered matrix.  
Let
\begin{equation}\label{eq_def_F_pm}
F_+:=\left\{ \left( \begin{smallmatrix}
u \\ 0    
\end{smallmatrix} \right)\ | u\in \C^{N} \right\} \subset \C^{2N} \ \mbox{ and }\   F_-:=\left\{ \left( \begin{smallmatrix}
0 \\ v    
\end{smallmatrix} \right)\ | v\in \C^{N} \right\} \subset \C^{2N}
\end{equation}
and for any $n\in \Z$,
\begin{equation}
 F_n:=\left\{ \left( \begin{smallmatrix}
u \\ v    
\end{smallmatrix} \right)\in  \C^{2N}\ \big| \ u=-\Phi_+^\downarrow(\ell n)(\Phi_+^\uparrow(\ell n))^{-1}  v \right\}.  
\end{equation}

Then, set
\begin{equation}\label{eq_def_Omega_epsilon}
     \Omega_\epsilon(n):= \Omega_\epsilon^{F_n}[(T_{\ell L}^{ \ell n}(E))^*] \cap \Omega_\epsilon^{F_+}[(T_{\ell L}^{ \ell n}(E))^*] \cap \Omega_\epsilon^{F_+}[T_{\ell L}^{ \ell n}(E)] \cap \Omega_\epsilon^{F_-}[(T_{\ell L}^{ \ell n}(E))^*] \cap \Omega_\epsilon^{F_-}[T_{\ell L}^{ \ell n}(E)].
\end{equation}
Note that if the transfer matrices are in $\SpC$, the vector subspace $F_n$ is $J$-Lagrangian, \emph{i.e.} it is orthogonal to itself for $J$ and of dimension $N$. If the transfer matrices are in $\SpONN$, $F_n$ is $(J,S)$-Lagrangian, \emph{i.e.} it is orthogonal to itself for $J$ and for $S$ and it is of dimension $N$. Also note that $F_+$ and $F_-$ are always $J$-Lagrangian, but they are not $(J,S)$-Lagrangian. To bypass this difficulty in the case of $N=2d$ even, we write $F_{+}=F_+^{(+)}\oplus F_+^{(-)}$ and $F_{-}=F_-^{(+)}\oplus F_-^{(-)}$ where, 
\begin{equation}\label{def_Fpm+-}
F_+^{(\pm)}:=\left\{ \left(\begin{smallmatrix}
x_1\\ \pm x_1 \\ \vdots\\ x_d \\ \pm x_d \\0 \\ \vdots \\0    
\end{smallmatrix} \right)\ ;\ (x_1,\ldots ,x_d)\in \C^d \right\} \ \mbox{ and }\   F_-^{(\pm)}:=\left\{ \left(\begin{smallmatrix}
 0 \\ \vdots \\0  \\  x_1\\ \pm x_1 \\ \vdots\\ x_d \\ \pm x_d
\end{smallmatrix} \right) \ ;\ (x_1,\ldots ,x_d)\in \C^d \right\}.
\end{equation}
Then all the $F_{\pm}^{(\pm)}$ are $(J,S)$-Lagrangian. 

One then define the event
\begin{align}
    \tilde{\Omega}_\epsilon(n):=  \Omega_\epsilon^{F_n}[(T_{\ell L}^{ \ell n}(E))^*] & \cap \Omega_\epsilon^{F_+^{(+)}}[(T_{\ell L}^{ \ell n}(E))^*] \cap \Omega_\epsilon^{F_+^{(+)}}[T_{\ell L}^{ \ell n}(E)]  \cap \Omega_\epsilon^{F_+^{(-)}}[(T_{\ell L}^{ \ell n}(E))^*] \cap \Omega_\epsilon^{F_+^{(-)}}[T_{\ell L}^{ \ell n}(E)]  \nonumber \\ 
    & \cap \Omega_\epsilon^{F_{-}^{(+)}}[(T_{\ell L}^{ \ell n}(E))^*] \cap \Omega_\epsilon^{F_-^{(+)}}[T_{\ell L}^{ \ell n}(E)]  \cap \Omega_\epsilon^{F_{-}^{(-)}}[(T_{\ell L}^{ \ell n}(E))^*] \cap \Omega_\epsilon^{F_-^{(-)}}[T_{\ell L}^{ \ell n}(E)].         \label{eq_def_Omega_epsilon_SpoNN}
\end{align}

\bigskip
We can control $||G_{\Lambda_{L}}^\omega(E,x,y)||$ on these two events.
\bigskip

\begin{prop}\label{bound:int_ker}
Let $I\subset \R$ a compact interval on which for every $E\in I$, $\inf_{E\in I}\gamma_N(E)>0$. Assume either that 
\begin{enumerate}
    \item $\forall E\in I$, $G(E)\in \SpC$ and $\omega \in \Omega_\epsilon:=\cap_{n\in [-L/3,L/3]}\Omega_\epsilon(n) $ or 
\item  $\forall E\in I$, $G(E)\in \SpONN$, $N$ is even and $\omega \in \tilde{\Omega}_\epsilon:=\cap_{n\in [-L/3,L/3]}\tilde{\Omega}_\epsilon(n).$
\end{enumerate}
Let us choose $\epsilon\in\left(0,\inf_{E\in I} \frac{\gamma_N(E)}{3}\right)$.    Then, there exists $C>0$ such that, for $L$ large enough, for all $E\in I$,  $x\in\supp \chi_{0,L/3} $ and $y \in\supp \Gamma_{0,L} $,
    \begin{equation}\label{ILSE_bound_G}
        ||G_{\Lambda_{L}}^\omega(E,x,y)||\leq C\e^{-2(\gamma_N(E)-\frac74\epsilon)\ell L}.
    \end{equation}
\end{prop}

\begin{proof} 
We give a complete proof in case (1).
Then at the end we will explain the slight differences in the proof when we assume 
that we are in case (2).

We will prove the proposition only for $y$ in $[\ell L-3\ell,\ell L-\ell]$. The proof for $y$ in $[-\ell L+\ell,-\ell L+3\ell]$ is similar. Since, in this case, $x<y$, the Green kernel will be given by the first line of \eqref{defGD}.

We begin with proving that on $\Omega_\epsilon$, for all $x\in[-\ell L/3,\ell L/3]$, the matrices $\Phi^\uparrow_+(x)$ and $\Phi_+^\downarrow(x)$ are invertible and we estimate their inverse. It corresponds to bounding from below their $N$-th singular value.The reader who is not familiar with inequalities on singular values can refer to Chapter~3 of~\cite{HornJohn}.

 By definition (recall~\eqref{eq_def_transfer_mat_xy}), for any $x$ in $\Lambda_L$,
 \[\begin{pmatrix}
  \Phi_+^\uparrow(x)\\\Phi_+^\downarrow(x)  
\end{pmatrix}=T_{\ell L}^{x }(E)\begin{pmatrix}
    0\\I_N
\end{pmatrix}.\]
We write the  singular value decomposition of $T_{\ell L}^{x}(E)$:
\begin{equation}
T_{\ell L}^{x}(E)=U^{(x)}\Sigma^{(x)} V^{(x)},\end{equation}
where $U^{(x)}$ and $V^{(x)}$ are unitary and, since we have a symplectic matrix, we can write $\Sigma^{(x)}=\left(\begin{smallmatrix}
    \Sigma_+^{(x)}&0\\0&\Sigma_-^{(x)}
\end{smallmatrix}\right)$ with $\Sigma_+^{(x)}=\diag(s_1(T_{\ell L}^{x}),\dots,s_N(T_{\ell L}^{x}))$ and $\Sigma_-^{(x)}=\diag(1/s_1(T_{\ell L}^{x}),\dots,1/s_N(T_{\ell L}^{x}))$ with $s_1(T_{\ell L}^{x})\geq\dots\geq s_N(T_{\ell L}^{x})\geq 1$. We can write a block decomposition for $U^{(x)}$ and $V^{(x)}$: $U^{(x)}=\left(\begin{smallmatrix}
    U_{11}^{(x)}&U_{12}^{(x)}\\U_{21}^{(x)}&U_{22}^{(x)}
\end{smallmatrix}\right)$ and $V^{(x)}=\left(\begin{smallmatrix}
    V_{11}^{(x)}&V_{12}^{(x)}\\V_{21}^{(x)}&V_{22}^{(x)}
\end{smallmatrix}\right)$. We have then
\[\begin{pmatrix}
  \Phi_+^\uparrow(x)\\\Phi_+^\downarrow(x)  
\end{pmatrix}=\begin{pmatrix}U_{11}^{(x)}\Sigma_+^{(x)}V_{12}^{(x)}+U_{12}^{(x)}\Sigma_-^{(x)}V_{22}^{(x)}\\U_{21}^{(x)}\Sigma_+^{(x)}V_{12}^{(x)}+U_{22}^{(x)}\Sigma_-^{(x)}V_{22}^{(x)}\end{pmatrix}.\]
As a consequence,  we can write that
\begin{align}
    s_N\left(\Phi_+^\uparrow(x)\right)&\geq s_N\left(U_{11}^{(x)}\Sigma_+^{(x)}V_{12}^{(x)}\right)-\|U_{12}^{(x)}\Sigma_-^{(x)}V_{22}^{(x)}\|\\
   &\geq s_N\left(U_{11}^{(x)}\right)s_N\left(\Sigma_+^{(x)}\right)s_N\left(V_{12}^{(x)}\right)-\|\Sigma_-^{(x)}\|,
\end{align}
where we used in the last line that  the blocks $U_{ij}^{(x)}$, $V_{ij}^{(x)}$ have norm less than 1.

We now estimate $s_N\left(\Sigma_+^{(x)}\right)$, which is also $\|\Sigma_-^{(x)}\|^{-1}$, on the event $\Omega_\epsilon$.
Let $n$ be the unique integer such that $x\in[(n-1)\ell, n\ell)$.
We first remark that, for all $p=1,\dots,2N$,
\begin{equation}\label{proche_entier}s_p\left(T_{\ell L}^{x}\right)\geq s_p\left(T_{\ell L}^{\ell n}\right)s_{2N}\left(T_{\ell n}^{x}\right)=s_p\left(T_{\ell L}^{\ell n}\right)\|T_{\ell n}^{x}\|^{-1}.\end{equation}
But we know from Lemma~\ref{lem:Gronwall1} that there exists a constant $C>0$, independent of $E$, $x$, $\omega$ and $L$, such that $\|T_{\ell n}^{x}\|\leq C$. As a consequence, $s_p\left(T_{\ell L}^{x}\right)\geq C^{-1}s_p\left(T_{\ell L}^{\ell n}\right)$. 
Consequently, on the event $\Omega_\epsilon(n)$,  $s_N\left(T_{\ell L}^{x}\right)\geq C^{-1}\e^{(\gamma_N(E)-\epsilon)\ell |n-L|}$ and $\|\Sigma_-^{(x)}\|\leq C\e^{-(\gamma_N(E)-\epsilon)\ell |n-L|}$, so
\begin{equation}
    s_N\left(\Phi_+^\uparrow(x)\right)\geq C^{-1}\e^{(\gamma_N(E)-\epsilon)\ell |n-L|}s_N\left(U_{11}^{(x)}\right)s_N\left(V_{12}^{(x)}\right)-C\e^{-(\gamma_N(E)-\epsilon)\ell |n-L|}
\end{equation}
and, similarly,
\begin{equation}
    s_N\left(\Phi_+^\downarrow(x)\right)\geq C^{-1}\e^{(\gamma_N(E)-\epsilon)\ell |n-L|}s_N\left(U_{21}^{(x)}\right)s_N\left(V_{12}^{(x)}\right)-C\e^{-(\gamma_N(E)-\epsilon)\ell |n-L|}.
\end{equation}

We are left with controlling $s_N\left(U_{11}^{(x)}\right)$, $s_N\left(U_{21}^{(x)}\right)$ and $s_N\left(V_{12}^{(x)}\right)$. Recalling the notation introduced in~\eqref{eq_def_F_pm}, we have that $U_{11}^{(x)}=\pi_{F_+}^*U^{(x)}\pi_{F_+}$, $U_{21}^{(x)}=\pi_{F_-}^*U^{(x)}\pi_{F_+}$ and $V_{12}^{(x)}=\pi_{F_+}^*V^{(x)}\pi_{F_-}$.
Since the sets $F_+$ and $F_-$ as defined in \eqref{eq_def_F_pm} are $J$-Lagrangian, we can prove as in Claim~3.4 and Remark~3.5 of~\cite{MaSo22} that, on $\Omega_\epsilon(n)$, we have for $L$ large enough (depending only on $\epsilon$ and $\ell$)
\begin{equation}
    s_N\left(V_{12}^{(x)}\right)\geq \e^{-\frac{\epsilon}{25}\ell|n-L|}\text{, } s_N\left(U_{21}^{(x)}\right)\geq \e^{-\frac{\epsilon}{25}\ell |n-L|}\text{ and } s_N\left(U_{11}^{(x)}\right)\geq \e^{-\frac{\epsilon}{25}\ell |n-L|}.
\end{equation}
More precisely, the first inequality corresponds to the event  $\Omega_\epsilon^{F_-}[T_{\ell L}^{ \ell n}(E)]$, the second to the event $\Omega_\epsilon^{F_-}[(T_{\ell L}^{ \ell n}(E))^*]$ and the third to $\Omega_\epsilon^{F_+}[(T_{\ell L}^{ \ell n}(E))^*]$.


As a consequence, on $\Omega_\epsilon(n)$, 
\begin{equation}\label{norm_centre}
    s_N\left(\Phi_+^{\uparrow\downarrow}(x)\right)\geq C^{-1} \e^{(\gamma_N(E)-\frac{27\epsilon}{25})\ell |n-L|}-C\e^{-(\gamma_N(E)-\epsilon)\ell |n-L|}.
\end{equation}
For $L$ large enough (depending here again only on $C$, $\epsilon$ and $\ell$), one gets on $\Omega_\epsilon(n)$, using that $|n-L|\geq 2L/3$,
\begin{equation}\label{norm_Phi_D}
     s_N\left(\Phi_+^{\uparrow\downarrow}(x)\right)\geq \e^{(\gamma_N(E)-2\epsilon) \tfrac{2\ell L}{3}} >0.
\end{equation}
In particular, $\Phi_+^\uparrow(x)$ and $\Phi_+^\downarrow(x)$ are invertible.

The next step to be able to apply Lemma~\ref{lem:Green_kerD} is to prove that $\Phi_+^\downarrow(x)\Phi_+^\uparrow(x)^{-1}-\Phi_-^\downarrow(x)\Phi_-^\uparrow(x)^{-1}$ and $\Phi_+^\uparrow(x)\Phi_+^\downarrow(x)^{-1}-\Phi_-^\uparrow(x)\Phi_-^\downarrow(x)^{-1}$ are invertible. By \eqref{bc_Phi_ILSE}, $\Phi_+^\downarrow(x)\Phi_+^\uparrow(x)^{-1}=(T_{\ell L}^x)_{22}((T_{\ell L}^x)_{12})^{-1} $ and $\Phi_-^\downarrow(x)\Phi_-^\uparrow(x)^{-1}= (T_{-\ell L}^x)_{22}((T_{-\ell L}^x)_{12})^{-1}$. It does not exactly correspond to the blocks involved in $X^+$ and $X^-$ in Equation~(33) of~\cite{MaSo22} since we did not use the same boundary conditions, but we will be able to carry on the same proof as \cite{MaSo22}. Since the sets $F_n$, $F_+$ and $F_-$ are $J$-Lagrangian, we prove as in Equation~(33) of~\cite{MaSo22} that there exists $C>0$ such that on $\Omega_\epsilon(n)$ 
\begin{equation}\label{norm_wron_D}s_N\left(\Phi_+^\downarrow(x)\Phi_+^\uparrow(x)^{-1}-\Phi_-^\downarrow(x)\Phi_-^\uparrow(x)^{-1}\right)\geq C^{-1} \e^{-\frac{\epsilon}{2}\ell L}
\end{equation} 
as well as the same result for $\Phi_+^\uparrow(x)\Phi_+^\downarrow(x)^{-1}-\Phi_-^\uparrow(x)\Phi_-^\downarrow(x)^{-1}$. Note that the constant $C>0$ appears as we estimate the $N$-th singular value of the matrix at point $x$ by its value at $n$  its integer part and the norm of the difference which is bounded in $L$ times $s_N(U_{21})$ which is lower bounded by an exponentially small term.  Note also that in the proof appears $s_N\left(V_{11}^{(x)}\right)$ which is controlled on $\Omega_\epsilon^{F_+}[T_{\ell L}^{ \ell n}(E)]$. 

Consequently, these two matrices are invertible and the norm of their inverse is smaller than $C \e^{\frac{\epsilon}{2}\ell L}$.

 Finally, we  remark that, for  $y\in[\ell L-3\ell,\ell L-\ell]$, we have by Lemma~\ref{lem:Gronwall1},
 \begin{equation}\label{norm_bord}
     ||\Phi_+^\uparrow(y)||^2+ ||\Phi_+^\downarrow(y)||^2\leq N\exp\left(2\int_{y}^{\ell L}|V_\omega(t)|\d t\right)\leq C,
 \end{equation}
 with $C$ independent of $E$, $\omega$ and $L$.

Together with \eqref{norm_Phi_D} and~\eqref{norm_wron_D}, it gives that on $\Omega_{\epsilon}$
\begin{equation}
    ||G_{\Lambda_L}^\omega(E,x,y)||\leq C\e^{-2(\gamma_N(E)-\frac{7}{4}\epsilon)\ell L},
\end{equation}
which concludes the proof of the proposition under the assumptions of Theorem \ref{thm:loc_criterion1}.

Now let us briefly explain how to adapt the proof when we are in case (2). First, remark that if the transfer matrices are in $\SpONN$, then $F_n$ is a $(J,S)$-Lagrangian since it depends on the transfer matrices through two of their blocks. Moreover, all the $F_{\pm}^{(\pm)}$ are $(J,S)$-Lagrangian.  Since, $\pi_{F_+}^*=\pi_{F_+^{(+)}}^* + \pi_{F_+^{(-)}}^*$, one has inequalities such as :
\begin{equation}\label{eq_sv_U_SpoNN}
s_N(U_{11}^{(x)})=s_N(\pi_{F_+} U^{(x)} \pi_{F_+}^{*}) =s_N(\pi_{F_+} U^{(x)} \pi_{F_+^{(+)}}^{*} +\pi_{F_+} U^{(x)} \pi_{F_+^{(-)}}^{*} )  \geq s_d (\pi_{F_+} U^{(x)} \pi_{F_+^{(+)}}^{*}) \geq   \e^{-\frac{\epsilon}{25}\ell |n-L|}.  
\end{equation}

The first inequality comes from the fact that all the singular values are of multiplicity $2$ using the decomposition $(4)$ of Proposition \ref{prop:decompG} hence $s_d(\pi_{F_+} U^{(x)} \pi_{F_+^{(+)}}^{*})$ is the $2d$-th singular value of $\pi_{F_+} U^{(x)} \pi_{F_+^{(+)}}^{*}$. The second inequality comes from the fact that $F_+^{(+)}$ is a subspace of $F_+$ for which we can use Claim~3.4 and Remark~3.5 of~\cite{MaSo22}. The same kind of inequalities allow to obtain a lower bound for $s_N(U_{21}^{(x)}-\Phi_+^\downarrow(x)(\Phi_+^\uparrow(x))^{-1} U_{11}^{(x)})$ which is needed to obtain a lower bound as \eqref{norm_wron_D} under the assumptions of case (2). Having in mind these specific properties when the transfer matrices are in $\SpONN$, one proves \eqref{ILSE_bound_G} in this case, the same way as it was done when they are in $\SpC$.
\end{proof}

In order to estimate the probability of $\Omega_{\epsilon}$ and $\tilde{\Omega}_{\epsilon}$, we prove a Large Deviation Property for the singular values of the products of transfer matrices. We need to do this separately for $\SpC$ and $\SpONN$.

\begin{prop}\label{prop_LDP_valeurssingulieres_SpC}

We fix a compact interval $I\subset \R$. We assume that for every $E\in I$:
\begin{enumerate}
    \item the Furstenberg  group $G(E)$ is included in $\SpC$  ;
    \item for every $p\in \{1,\ldots, N\}$, $G(E)$ is $J$-$L_p$-strongly irreducible.
\end{enumerate}
Then for all $\epsilon >0$ and all $E\in I$, there exist $C(\epsilon,E)>0$ and $c(\epsilon,E)>0$ such that, for all $ p\in\{1,\dots,N\}$, any $J$-Lagrangian subspace $F$ and all integers $m,n$,
\begin{equation}\label{eq_LDP_vs_1}
\PP\left(\left\{\left|\frac{1}{\ell(n-m)}\log s_p\left(T_{\ell m}^{\ell n}(E)\right)-\gamma_p(E)\right|\geq \epsilon\right\}\right) \leq C(\epsilon,E) \e^{-c(\epsilon,E) \ell |n-m|}
\end{equation}
and 
\begin{equation}\label{eq_LDP_vs_2}
\PP\left(\left\{\left|\frac{1}{\ell(n-m)}\log s_p\left(T_{\ell m}^{\ell n}(E)\pi_{F}^{*}\right)-\gamma_p(E)\right|\geq \epsilon\right\}\right) \leq C(\epsilon,E) \e^{-c(\epsilon,E) \ell |n-m|}. 
\end{equation}
\end{prop}

\begin{rem}\label{rem:LDP}
The constants $C(\epsilon,E)$ and $c(\epsilon,E)$ depend \emph{a priori} on $E$ and $\epsilon$ but they can be taken uniform in $\epsilon$ as it tends to $0$ and uniform in $E$ on the compact interval $I$. This is one of the reasons why we need to take the interval $I$ compact. 
\end{rem}

\begin{proof}
First of all, recall that $\SpC$ can be identified to  $\mathrm{Sp}_{2\mathrm{N}}(\R)$  using the following application which split the real and imaginery parts of the matrices in $\mathcal{M}_{\mathrm{2N}}(\C)$: 
$$\begin{array}{cll}
           \mathcal{M}_{\mathrm{2N}}(\C) & \to & \mathcal{M}_{\mathrm{4N}}(\R) \\[2mm]
	    A+iB & \mapsto & \left( \begin{smallmatrix}
				      A & -B \\
				      B & A
				      \end{smallmatrix} \right) .
				\end{array}$$
In this identification, all the multiplicities of the Lyapunov exponents and the singular values are doubled. Hence, we can use freely all the results concerning i.i.d. sequences of random matrices in the real symplectic group.

Recall that for a symplectic matrix $M\in\SpC$, $s_p(M^{-1})=s_p(M)$ for every $p\in \{1,\ldots 2N\}$. Hence one can assume that $m\leq n$ without loss of generality since $(T_{\ell m}^{\ell n}(E))^{-1}=T_{\ell n}^{\ell m}(E)$.

Let  $ p\in\{1,\dots,N\}$. As in the proof of Proposition \ref{prop:sepLyap}, for each $M\in \SpN$, we denote by $\widehat{M}$ the matrix of $\GL_k(\R)$ such that
 \begin{equation}
  \widehat{M}_{ij}=\langle f_i, \Lambda^{p}Mf_j\rangle.
 \end{equation}
where $k$ is the dimension of $J\mbox{-}L_{p}$ and $(f_1,\dots,f_k)$ is an orthonormal basis of $J$-$L_{p}$, with $f_1=e_1\wedge\cdots\wedge e_p$. Let us now denote by $\widehat{G(E)}$ the subgroup of $\GL_k(\R)$ which is generated by the matrices $\widehat{M}$ for $M\in G(E)$. Then, since $G(E)$ is $J$-$L_p$-strongly irreducible,   $\widehat{G(E)}$ is strongly irreducible. Hence, applying \cite[Theorem A.V.6.2]{BL}, one gets the existence of $\alpha>0$ such that for any $\epsilon >0$ and any $\bar{x}\in \mathsf{P}(J\mbox{-}L_p)$,
\begin{equation}\label{eq_LDP_proof_1}
\limsup_{|n-m|\to +\infty} \frac{1}{\ell|n-m|} \log\ \mathbb{P}\left( \left| \frac{1}{\ell|n-m|}\log(||\Lambda^p T_{\ell m}^{\ell n}(E) \bar{x}||)- (\gamma_1+\cdots+\gamma_p)(E) \right| >  \epsilon\right) \leq -\alpha
\end{equation} 
and 
\begin{equation}\label{eq_LDP_proof_2}
 \limsup_{|n-m|\to +\infty} \frac{1}{\ell|n-m|} \log\ \mathbb{P}\left(  \left| \frac{1}{\ell|n-m|}\log(||\Lambda^p T_{\ell m}^{\ell n}(E)||)- (\gamma_1+\cdots+\gamma_p)(E) \right| >  \epsilon\right) \leq -\alpha.   
\end{equation}

Indeed, since the function  $x\mapsto V_{\omega}^{(n)} (x)$ is  uniformly bounded in $x$, $n$ and $\omega$, and since the support of the common law of the transfer matrices is bounded, the assumption of finiteness of the integral in \cite[Theorem A.V.6.2]{BL} is satisfied. 

Now, let us take $F$ a $J$-Lagrangian subspace of $\C^{2N}$. Then,
\begin{align*}
|| \Lambda^p (T_{\ell m}^{\ell n}(E) \pi_F^*)|| & = \sup_{\substack{u_1\wedge\cdots \wedge u_p \in \mathsf{P}(J\mbox{-}L_p) \\ u_i \in F} } ||(\Lambda^p T_{\ell m}^{\ell n}(E))(u_1\wedge\cdots \wedge u_p)|| \\
 & = || \Lambda^p T_{\ell m}^{\ell n}(E) \ \bar{u}||,\ \mbox{ for some } \bar{u}\in  \mathsf{P}(J\mbox{-}L_p)
\end{align*}
since the supremum is attained by compactness of $\mathsf{P}(J\mbox{-}L_p)$. Hence, \eqref{eq_LDP_proof_1} rewrites,
\begin{equation}\label{eq_LDP_proof_3}
\limsup_{|n-m|\to +\infty} \frac{1}{\ell|n-m|} \log\ \mathbb{P}\left( \left| \frac{1}{\ell|n-m|}\log(||\Lambda^p (T_{\ell m}^{\ell n}(E) \pi_F^*)||)- (\gamma_1+\cdots+\gamma_p)(E) \right| >  \epsilon\right) \leq -\alpha
\end{equation} 
for any $J$-Lagrangian $F$. Let, for $n,m\in \Z$, $p\in \{1,\ldots,N\}$, $\epsilon>0$ and $F$ a $J$-Lagrangian,
$$A_{n,m,p}(\epsilon,F)=\left\{ \left| \frac{1}{\ell|n-m|}\log(||\Lambda^p (T_{\ell m}^{\ell n}(E) \pi_F^*) ||)- (\gamma_1+\cdots+\gamma_p)(E) \right| >  \epsilon  \right\},$$

$$A_{n,m,p}(\epsilon)=\left\{ \left| \frac{1}{\ell|n-m|}\log(||\Lambda^p T_{\ell m}^{\ell n}(E)||)- (\gamma_1+\cdots+\gamma_p)(E) \right| >  \epsilon\right\},$$

$$B_{n,m,p}(\epsilon,F)=\left\{ \left| \frac{1}{\ell|n-m|}\left(\log(||\Lambda^p (T_{\ell m}^{\ell n}(E) \pi_F^*)||)- \log(||\Lambda^{p-1} (T_{\ell m}^{\ell n}(E) \pi_F^*)||) \right) - \gamma_p(E) \right| >  \epsilon\right\}$$
and
$$B_{n,m,p}(\epsilon)=\left\{ \left| \frac{1}{\ell|n-m|}\left( \log(||\Lambda^p T_{\ell m}^{\ell n}(E)||)- \log(||\Lambda^{p-1} T_{\ell m}^{\ell n}(E) ||) \right)- \gamma_p(E) \right| >  \epsilon\right\}.$$
Then, one has
\begin{equation}\label{eq_LDP_proof_4}
B_{n,m,p}(2\epsilon) \subset A_{n,m,p}(\epsilon) \cap A_{n,m,p-1}(\epsilon) \ \mbox{ and } B_{n,m,p}(2\epsilon,F) \subset A_{n,m,p}(\epsilon,F) \cap A_{n,m,p-1}(\epsilon,F).
\end{equation}

Since for any $p\in \{1,\ldots ,N\}$, $||\Lambda^p T_{\ell m}^{\ell n}(E) ||=s_1(T_{\ell m}^{\ell n}(E))\cdots s_p(T_{\ell m}^{\ell n}(E))$, combining  \eqref{eq_LDP_proof_4} and \eqref{eq_LDP_proof_2} one gets \eqref{eq_LDP_vs_1}. One also has, for any $J$-Lagrangian $F$,  $||\Lambda^p (T_{\ell m}^{\ell n}(E) \pi_F^*) ||=s_1(T_{\ell m}^{\ell n}(E) \pi_F^*)\cdots s_p(T_{\ell m}^{\ell n}(E) \pi_F^*)$. Hence, combining  \eqref{eq_LDP_proof_4} and \eqref{eq_LDP_proof_3} one gets \eqref{eq_LDP_vs_2}. This achieves the proof.
\end{proof}

We then prove the counterpart of Proposition \ref{prop_LDP_valeurssingulieres_SpC} in the case covered by Theorem \ref{thm:loc_criterion2}.

\begin{prop}\label{prop_LDP_valeurssingulieres_SpONN}

We fix a compact interval $I\subset \R$. We assume that for every $E\in I$:
\begin{enumerate}
    \item the Furstenberg  group $G(E)$ is included in $\SpONN$  ;
    \item for every $2p\in \{1,\ldots, N\}$, $G(E)$ is $(J,S)$-$L_{2p}$-strongly irreducible.
\end{enumerate}
Then for all $\epsilon >0$ and all $E\in I$, there exist $C(\epsilon,E)>0$ and $c(\epsilon,E)>0$ such that, for all $ p\in\{1,\dots,N\}$, any $(J,S)$-Lagrangian subspace $F$ and all integers $m,n$,
\begin{equation}\label{eq_LDP_vs_3}
\PP\left(\left\{\left|\frac{1}{\ell|n-m|}\log s_{p}\left(T_{\ell m}^{\ell n}(E)\right)-\gamma_p(E)\right|\geq \epsilon\right\}\right) \leq C(\epsilon,E) \e^{-c(\epsilon,E) \ell |n-m|}
\end{equation}
and 
\begin{equation}\label{eq_LDP_vs_4}
\PP\left(\left\{\left|\frac{1}{\ell|n-m|}\log s_p\left(T_{\ell m}^{\ell n}(E)\pi_{F}^{*}\right)-\gamma_p(E)\right|\geq \epsilon\right\}\right) \leq C(\epsilon,E) \e^{-c(\epsilon,E) \ell |n-m|}. 
\end{equation}
\end{prop}

\begin{proof}
The proof is very similar to the proof of Proposition \ref{prop_LDP_valeurssingulieres_SpC}. First, we introduce $\widetilde{G(E)}$ defined as in the proof of Proposition \ref{prop:sepLyap}. Then, since $G(E)$ is $(J,S)$-$L_{2p}$-strongly irreducible,   $\widehat{G(E)}$ is strongly irreducible and we can again apply  \cite[Theorem A.V.6.2]{BL} to get for each $2p\in \{1,\ldots,N\}$, the analogous of \eqref{eq_LDP_proof_2} for $\Lambda^{2p}$ instead of $\Lambda^p$ and the same for \eqref{eq_LDP_proof_1} for any $\bar{x}\in \mathsf{P}((J,S)\mbox{-}L_{2p})$. Hence, one gets \eqref{eq_LDP_proof_3} for any $(J,S)$-Lagrangian $F$ and for $2p$ instead of $p$.
Then, one has, for  any $n,m\in \Z$, $2p\in \{1,\ldots,N\}$, $\epsilon>0$ and $F$ a $(J,S)$-Lagrangian,
\begin{equation}\label{eq_LDP_proof_5}
B_{n,m,2p}(2\epsilon) \subset A_{n,m,2p}(\epsilon) \cap A_{n,m,2p-2}(\epsilon) \ \mbox{ and } B_{n,m,2p}(2\epsilon,F) \subset A_{n,m,2p}(\epsilon,F) \cap A_{n,m,2p-2}(\epsilon,F).
\end{equation}
One conclude as in the proof of Proposition \ref{prop_LDP_valeurssingulieres_SpC} using that for each $2p\in \{1,\ldots, N\}$ and each $(J,S)$-Lagrangian $F$,  $s_{2p}(T_{\ell m}^{\ell n}(E))=s_{2p+1}(T_{\ell m}^{\ell n}(E))$, $s_{2p}(T_{\ell m}^{\ell n}(E)\pi_{F}^{*})=s_{2p+1}(T_{\ell m}^{\ell n}(E)\pi_{F}^{*})$ and $\gamma_{2p}(E)=\gamma_ {2p+1}(E)$ since the transfer matrices are in $\SpONN$.
\end{proof}

\begin{proof}[Proof of the ILSE \eqref{eq:ILSEstrong}]
Assume either the assumptions of Theorem \ref{thm:loc_criterion1} or those of  Theorem \ref{thm:loc_criterion2}. Applying Proposition \ref{prop_LDP_valeurssingulieres_SpC} for the $J$-Lagrangian spaces $F_n$, $F_+$ and $F_-$ one gets  that there exist $C(\epsilon,E)>0$ and $c(\epsilon,E)>0$ such that,
\begin{equation}\label{Estimee_proba_omega_epsilon_n}
\forall \epsilon >0,\ \forall n\in \Lambda_{\frac{L}{3}},\ \PP\left( \Omega_{\epsilon }(n)\right) \leq C(\epsilon,E) \e^{-c(\epsilon,E) \ell |n-L|}
\end{equation}
from which we deduce, since there is only a finite number of real numbers $n \ell$ in $\Lambda_{\frac{L}{3}}$, that there exist $C>0$ and $c>0$ such that,
\begin{equation}\label{Estimee_proba_omega_epsilon}
\forall \epsilon >0,\ \PP\left( \Omega_{\epsilon }\right) \leq C \e^{-c \frac43 \ell L}
\end{equation}
remembering that we can choose the constants $C(\epsilon,E)$ and $c(\epsilon,E)$ uniform in $\epsilon$ when it tends to $0$ and uniform in $E\in I$ by compactness of $I$.

Applying Proposition \ref{prop_LDP_valeurssingulieres_SpONN} for the $(J,S)$-Lagrangian spaces $F_n$, $F_+^{\pm}$ and $F_-^{\pm}$ one gets also that there exist $C>0$ and $c>0$ such that,
\begin{equation}\label{Estimee_proba_tilde_omega_epsilon}
\forall \epsilon >0,\ \PP\left( \tilde{\Omega}_{\epsilon }\right) \leq C \e^{-c \frac43 \ell L}
\end{equation}

Proposition~\ref{bound:int_ker} together with~\eqref{eq:schur} and \eqref{Estimee_proba_omega_epsilon} or \eqref{Estimee_proba_tilde_omega_epsilon} (up to a passage to the complementary sets of $\Omega_{\epsilon }$ and $\tilde{\Omega}_{\epsilon }$) lead to~\eqref{eq:ILSEstrong} under either assumptions of Theorem \ref{thm:loc_criterion1} or Theorem \ref{thm:loc_criterion2}.
\end{proof}


\begin{proof}[Proof of Theorem~\ref{thm:loc_criterion1} and Theorem \ref{thm:loc_criterion2}]
In view of Theorems~\ref{pi} and~\ref{thm:main-section4}, proving Theorem~\ref{thm:loc_criterion1} and Theorem \ref{thm:loc_criterion2} reduces to prove  all the hypotheses stated in Section \ref{subsec:MSA} on interval $I$. It is easy to see that the restriction of the operators to intervals with Dirichlet boundary conditions  makes $\{\DO \}_{\omega \in \Omega}$ a  standard family of operators.

As in~\cite{Z23}, we first remark that, for all families of operators of the form given by \eqref{def_op_D_general}, (SGEE), (SLI), (IAD) and (EDI) have already been proven in \cite[Proof of Theorem 4.1]{BCZ} on $\R$. Similarly, (NE) is proven in the same paper, in the proof of Theorem~4.2. Even if there are extra hypotheses in that paper, they are not used in the proof of these specific assumptions. We already proven the Wegner estimate (W) under the hypotheses of Theorem~\ref{thm:loc_criterion1} or Theorem \ref{thm:loc_criterion2}.

 Finally, we have have just proven in this Section the ILSE under both hypotheses of Theorem~\ref{thm:loc_criterion1} or Theorem \ref{thm:loc_criterion2}, which achieves the proof of these two theorems.
\end{proof}

\section{Application for potentials splitting in a sum of two Pauli matrices}\label{sec:splitting:Pauli}



The goal of this section is  to establish the properties of Lyapunov exponents  for the models introduced in Section~\ref{sec:exemples}. We prove that the hypotheses of either Theorem~\ref{thm:loc_criterion1}, \ref{thm:loc_criterion2}  or~\ref{thm:deloc_crit} are satisfied, which respectively leads to the localization and delocalization results given by Theorem~\ref{thm:localization:exemple}.  The fact that the hypotheses of the localization criterion are satisfied depends in fact not really on the family of measures $\mu_E$, but only on the Furstenberg groups $G(E)$. Indeed, the Furstenberg  group can be independent of $E$ even if the measure $\mu_E$ depends on $E$.

Therefore, our first task is to determine the Furstenberg groups in the five different cases.
Since all transfer matrices are symplectic and have real elements, $G(E)$ is always included in the real symplectic group $\SpN$.  In some of the cases, we can have more. Let $\Delta$, the tridiagonal matrix with 0 on the diagonal and 1 on the upper and lower diagonal.
\begin{prop}\label{prop:inclusFurst}
    In all cases, $G(E)\subset\SpN$. Moreover, 
    \begin{itemize}
        \item in Case~1, $G(E)\subset \SO_{\mathrm{N}}(\R)$;
        \item in Case~5, if $V_{\mathrm{per}}=\Delta$, $G(E)\subset\SpONN$.
    \end{itemize}
\end{prop}
\begin{proof}
We have already proved in Section~\ref{sec:trans_mat} that the Furstenberg group is always included in $\SpN$. In some of the cases defined in Table~\ref{tab:cas}, we can prove stronger inclusions. First, we have that
\[\frac{\d}{\d y}\left[(T_x^y(E))^*T_x^y(E)\right]=(T_x^y(E))^*[V_\omega+V_{\mathrm{per}},J]T_x^y(E).\]
When $V_\omega+V_{\mathrm{per}}=\sigma_0\otimes V_0$, \emph{i.e.} in case 1, $[V_\omega+V_{\mathrm{per}},J]=0$ so $(T_x^y(E))^*T_x^y(E)=I_{2N}$ for all $x$ and $y$: $G(E) \subset \SO_{2N}(\R)$.

If $V_{\mathrm{per}}=\Delta$, in Case 5, a direct computation shows that for any $x,y\in \R$ and any $E\in \R$, $T_x^y(E)\in \SpONN$ hence the result on $G(E)$. 
\end{proof}

This has already consequences on the Lyapunov exponents in some cases. We prove that some Lyapunov exponents vanish.

\begin{thm}\label{thm:zeroLyap}
    In Case 1, all Lyapunov exponents are identically 0.

    In Case 5 for odd $N$ and $V_{\mathrm{per}}=\Delta$, the lowest Lyapunov exponent $\gamma_N$ is identically 0.
\end{thm}
\begin{proof}
     The fact that all Lyapunov exponents associated with a sequence of matrices included in the compact group  $\mathrm{SO}_{2N}(\R)$ are  identically zero  is a direct consequence of Proposition~A.III.5.6 of~\cite{BL}.

The result for Case 5 when $N$ is odd, is a consequence of Proposition~\ref{degen_gamma} combined with $\gamma_p=-\gamma_{2N-p+1}$ for all $p$, which leads to $\gamma_N=\gamma_{N+1}=0$.
\end{proof}
From Theorems~\ref{thm:zeroLyap} and~\ref{thm:deloc_crit}, we directly get the  \emph{delocalization} results in Theorem~\ref{thm:localization:exemple}.

To get localization results, we need more information.
In Section~\ref{Liealg}, we compute the  Furstenberg groups in the different case when  $\ell$ is close to 0, which means that the system is strongly disordered.
\begin{thm}\label{thm:Furst_groups}
   Recall the different cases given in Table~\ref{tab:cas}. In case 1, $G(E)$ is a subgroup of $\mathrm{SO}_{2N}(\R)$.   Moreover, there exists $\ell_C>0$ such that, for all $\ell\in(0,\ell_C)$, there exists a compact interval $I(N,\ell)$ such that for all $E\in I(N,\ell)\setminus\{0\}$,
    \begin{itemize}
     \item in cases 2, 3 and 4, $G(E)=\SpN$.
     \item in Case~5,  if $V_{\mathrm{per}}=\Delta$, $G(E)=\SpONN$.
    \end{itemize}
\end{thm}

We then get that the hypotheses of either Theorem~\ref{thm:loc_criterion1} or Theorem~\ref{thm:loc_criterion2} are satisfied.
\begin{thm}\label{thm:sep_Lyap}
Fix $\ell<\ell_C$, for the $\ell_C$ introduced in Theorem~\ref{thm:Furst_groups}.
 Then, for all $E\in I(N,\ell,V)\setminus\{0\}$,
\begin{itemize}
\item in case 2, 3 or 4, $G(E)$ is $p$-contracting and $J$-$L_p$-strongly irreducible for all  $p$ in $\{1,\dots,N\}$.
\item in case 5,  if $V_{\mathrm{per}}=\Delta$, $G(E)$ is $2p$-contracting and $(J,S)$-$L_{2p}$-strongly irreducible for every $2p$ in $\{1,\dots,N\}$.
\end{itemize}
\end{thm}
\begin{proof}

    The case where $G(E)=\SpN$ is treated in the proof of Proposition~A.IV.3.5 of~\cite{BL}.

    The remaining case is studied in section~\ref{sec:cas5} of the present paper.
\end{proof}
Combining Theorems~\ref{thm:sep_Lyap} and~\ref{thm:loc_criterion1} or~\ref{thm:loc_criterion2}, we get the \emph{localization} results of Theorem~\ref{thm:localization:exemple}, at least for $V_{\mathrm{per}}=\Delta$. We give the complete proof of  Theorem~\ref{thm:localization:exemple} in Section \ref{subsec:proof-loc-exemple}.

\subsection{Determination of the Furstenberg groups}\label{Liealg}


In this section, we prove Theorem~\ref{thm:Furst_groups}.
 For models of the type introduced in Section~\ref{sec:exemples}, the common distribution of the $T_{\omega^{(n)}}(E)$'s is $\mu_E:=(T_{\omega^{(0)}}(E))_{*}\, (\nu)$ and we have the internal description of $G(E)$:
\begin{equation}\label{eq_GE_intern}
G(E)=\overline{<T_{\omega^{(0)}} (E)\ ;\ \omega^{(0)} \in \supp \nu >}\ \mbox{for all}\ E\in \R.
\end{equation}

Since we consider potentials which are constant on each interval $(n\ell, (n+1)\ell)$,
one has: 
\begin{equation}\label{mat_transfer_explicit_Xn}
\forall \ell>0,\ \forall n \in \Z,\ \forall E\in \R,\ T_{\omega^{(n)}}(E) = \exp\left( \ell  X_{\omega^{(n)}}(E) \right),
\end{equation}
where, using the decomposition~\eqref{decom:Vper}--\eqref{decomp:Vomega},
\begin{equation}\label{def_Xn}
\begin{split}
\forall E \in \R,\ \forall n\in \Z,\  X_{\omega^{(n)}}(E)&:= J(V_{\mathrm{per}}+V_{\omega^{(n)}}-E)\\
&=\left( \begin{smallmatrix}
-\alpha_1 \hat{V}_{\mathrm{per}} - \beta_1 \hat{V}_{\omega^{(n)}} & E-(\alpha_0-\alpha_3)\hat{V}_{\mathrm{per}} - (\beta_0-\beta_3)  \hat{V}_{\omega^{(n)}} \\
-E+(\alpha_0+\alpha_3)\hat{V}_{\mathrm{per}} + (\beta_0+\beta_3)  \hat{V}_{\omega^{(n)}} & \alpha_1 V + \beta_1 \hat{V}_{\omega^{(n)}}
\end{smallmatrix}\right)
\end{split}
\end{equation}

The explicit formula for transfer matrices, given in~\eqref{mat_transfer_explicit_Xn}, is not easy to handle since it is a matrix exponential in which the random parameters and the energy parameter are mixed. As a consequence, we would like to determine the Furstenberg groups directly from the matrices $X_{\omega^{(n)}}(E)$. This is possible when the disorder is large enough thanks to the following
 result due to  Breuillard and Gelander \cite{BG03}.

\begin{thm}[\cite{BG03}, Theorem 2.1]\label{thm_breuillard}
Let $G$ be a real, connected, semisimple Lie group, whose Lie algebra is $\mathfrak{g}$. Then, there is a neighborhood $\tilde{\mathcal{O}}$ of $1$ in $G$, on which $\log=\exp^{-1}$ is a well defined diffeomorphism, such that $g_{1},\ldots,g_{m}\in \tilde{\mathcal{O}}$ generate a dense subgroup whenever $\log g_{1},\ldots,\log g_{m}$ generate $\mathfrak{g}$.
\end{thm}
Note that this result holds only when we consider a finite number of generators. For any $(i,j)$ in $\{1,\dots,N\}$ we denote by $E_{ij}$  the $N$-by-$N$ matrix with a coefficient 1 in position $(i,j)$ and 0 everywhere else. For any $P\in\{0,1\}^N$, we denote by $X_J(E)$ the matrix $X_{\omega^{(0)}}(E)$ for $\hat{V}_{\omega^{(0)}}=\sum_{P_i=1}E_{ii}$. Let us introduce
\[G_{\{0,1\}}(E):=<\exp\left(\ell X_P (E)\right)\ ;\ P\in\{0,1\}^N >\subset G(E)\] since 
$\{0,1\}^N\subset \supp \nu$. The group $G_{\{0,1\}}(E)$ has $2^N$ generators. We will use Theorem~\ref{thm_breuillard} to prove that $G_{\{0,1\}}(E)$ is dense in $\SpN$ or $\SpONN$, depending on the case. Together with Proposition~\ref{prop:inclusFurst}, it will conclude the proof of Theorem~\ref{thm:Furst_groups}.

We first need to ensure that the groups $\SpN$ and $\SpONN$ are connected and semisimple. It is well-known that the symplectic group has these properties. In Appendix~\ref{app:G5}, we prove that they hold for the group~$\SpONN$ as well.


We next prove that, under good conditions on $E$ and $\ell$, the transfer matrices are in an appropriate neighborhood of the unit. We prove it in a general result holding for all operators of the form given by~\eqref{def_op_D_general} such that the  potential $V_\omega^{(n)}$ is constant on $[0,\ell]$ and has values in the space of real symmetric matrices.
\begin{lemma}\label{lem:normX}
    For all neighborhood $\mathcal{O}$ of $\mathrm{I}_{\mathrm{2N}}$ in $\SpN$, there exists $\ell_C>0$, $c_1,c_2\in \R$ and $d>0$ such that if $\ell< \ell_C$, then $c_1-\frac{d}{\ell}<c_2+\frac{d}{\ell}$ and $T_{\omO}(E)\in\mathcal{O}$ for all $\omO\in\{0,1 \}^N$ and all $E\in[c_1-\frac{d}{\ell},c_2+\frac{d}{\ell}]$.
\end{lemma}
Note that we do not need to prove an analogous result for $\SpONN$
 since, for any neighborhood $\mathcal{O}$ of $\mathrm{I}_{\mathrm{2N}}$ in $\SpONN$, there exists a neighborhood $\mathcal{O}'$ of $\mathrm{I}_{\mathrm{2N}}$ in $\SpN$ such that $\mathcal{O}=\mathcal{O}'\cap \SpONN.$
\begin{proof}
 Given a neighborhood $\mathcal{O}$ of 1 in $\SpN$, we set: 
$$\dlO:=\sup\{ R>0\ |\ \exp\left[\overline{B}(0,R)\right]\subset\mathcal{O} \},$$
where $\overline{B}(0,R)$ is the closed ball, centered on $0$ and of radius $R>0$, for the topology induced on the Lie algebra $\spN$  by the matrix norm induced by the euclidean norm on $\R^{2N}$. Then, $T_{\omO}(E)$ will be in $\mathcal{O}$ as soon as $\ell\|X_{\omO}(E)\|\leq \dlO$.

We remark that $X_{\omO}(E)=J^{-1}(E-V_{\omO})$. Since $J^{-1}$ is an orthogonal matrix, we have, denoting by $(\lambda_i^{\omO})_{1\leq i \leq 2N}$ the eigenvalues of the symmetric matrix of $V_{\omO}$, \[\|X_{\omO}(E)\|=\|E-V_{\omO}\|=\max_{1\leq i\leq 2N}|\lambda_i^{\omO}-E|.\]

We want to find an interval of values of $E$ on which  $\ell \|X_{\omO}(E)\|<\dlO$. In other words, we want to characterize the set
\[I(\ell,\mathcal{O}):=\left\{E\in\R, \max_{\omega\in\{0,1\}^N}\max_{1\leq i\leq N}|\lambda_i^\omega-E|\leq\frac{\dlO}{\ell}\right\}.\]
We see that 
\[I(\ell,\mathcal{O})=\bigcap_{\omega\in\{0,1\}^N}\bigcap_{1\leq i\leq N}\left[\lambda_i^\omega-\frac{\dlO}{\ell},\lambda_i^\omega+\frac{\dlO}{\ell}\right].\]

We define 
\[\lambda_{\max}:=\max_{\omega\in\{0,1\}^N}\max_{1\leq i\leq N}\lambda_i^\omega\text{ and }\lambda_{\min}:=\min_{\omega\in\{0,1\}^N}\min_{1\leq i\leq N}\lambda_i^\omega.\]
If $(\lambda_{\max}-\lambda_{\min})/2<\dlO/\ell$, then $I(\ell, \mathcal{O})=[\lambda_{\max}-\dlO/\ell,\lambda_{\min}+\dlO/\ell]$ which is a nonempty interval.
As a consequence, we have proved that if $\ell<\ell_C:=2\dlO/(\lambda_{\max}-\lambda_{\min})$, then for all $E\in I(\ell,N)$ $\ell\|X_{\omO}(E)\|<\dlO$.
As a consequence, $T_{\omO}(E)=\exp(\ell X_\omO(E))\in \mathcal{O}$. 
\end{proof}

Finally, we prove that, in each case, the Lie algebra generated by the $X_{\omega^{(0)}}(E)$'s corresponds with the Lie algebra associated with the group appearing in Theorem~\ref{thm:Furst_groups}.
Let us denote, for every $E\in \R$,
\begin{equation}\label{def_gE}
\mathfrak{g}(E): = \mathrm{Lie}\left\{ X_P(E) ;\ P\in\{0,1\}^N\right\}.
\end{equation}

We explicitly compute $\gg(E)$ in the different cases described in Table~\ref{tab:cas}. In order to state the result, we introduce the following algebras:
\begin{itemize}
    \item $\spN$ is the algebra associated with the symplectic group $\SpN$, it consists of all matrices of the shape $\begin{pmatrix}
        A&B\\C &-^tA
    \end{pmatrix}$, where $B$ and $C$ are symmetric.
    \item Let us introduce the operator $s$ defined on $\mathcal{M}_{N}(\R)$ by
$(^sM)_{ij}=(-1)^{i-j+1}M_{ji}$. 
 We define the Lie algebra 
 \begin{equation}\label{eq_expr_sponn}
\spONN:= \left\{\begin{pmatrix}
     A&B\\^sB&-^tA
    \end{pmatrix}
\text{ with } ^sA=A, ^tB=B\right\}.
\end{equation} This is the Lie algebra associated with the group $\SpONN$, as we prove in Appendix~\ref{app:G5}.
    \end{itemize}
We have the following result.
\begin{prop}\label{prop:Lie_alg}
Let us recall the 5 cases defined in Table~\ref{tab:cas} and assume that $V_{\mathrm{per}}=\Delta$. Then, for all $E\neq0$,
\begin{itemize}
    \item in cases 2, 3 and 4, $\mathfrak{g}(E)=\spN$; 
    \item in case 5, $\gg(E)=\spONN$.
\end{itemize}
\end{prop}
In the rest of the section, we prove this proposition in the different cases. We fix $E\neq 0$.
We use the following lemma.
\begin{lemma}\label{lem:struc_spN}
 We introduce the following notations:
 \begin{equation}\label{def:XijYij}
Z_{ij}:=\begin{pmatrix}
        0&E_{ij}+E_{ji}\\
        E_{ij}+E_{ji}&0
       \end{pmatrix}
       \quad
Y_{ij}:=\begin{pmatrix}
        0&E_{ij}+E_{ji}\\
        -E_{ij}-E_{ji}&0
       \end{pmatrix}.
       \end{equation}
If some subalgebra $\gg$ of $\spN$ contains all the $Z_{ij}$'s and $Y_{ij}$'s for $|i-j|\leq 1$ , then $\gg=\spN$.
  \end{lemma}
  A proof of this result can be found in the proof of Lemma~1 of~\cite{Bou09}, even if it is not really stated.
  



\subsubsection*{Case 2: potentials both on $\sigma_3$}
Recall that, according to Table~\ref{tab:cas}, this case covers the case where potentials are both on $\sigma_1$ as well. We consider the family of matrices $\{X_P(E)\}_{P\in \{0,1\}^N}$, as defined in~\eqref{def_Xn}, for $\alpha_3=\beta_3=1$ and all other coefficients equal to 0. We have\[X_P(E)=\left( \begin{smallmatrix}
0 & E+\Delta +\sum_{P_i=1} E_{ii} \\
-E+\Delta +  \sum_{P_i=1} E_{ii} & 0
\end{smallmatrix} \right).\]

In view of Lemma~\ref{lem:struc_spN}, we will prove that, for such a family $\{X_P(E)\}_{P\in\{0,1\}^N}$,  all  matrices $Z_{ij}$ and $Y_{ij}$ for $|i-j|\leq 1$ are in $\gg(E)$.
First, for any $P,P'\in \{0,1\}^N$, \[X_P(E)-X_{P'}(E)=\left(\begin{smallmatrix}                                                                        0&\sum_{P_i=1} E_{ii}-\sum_{P'_i=1} E_{ii}\\                                                                        \sum_{P_i=1} E_{ii}-\sum_{P'_i=1} E_{ii}&0                                                                                     \end{smallmatrix}\right).\] 
For all $i\in \{1,\dots,N\}$, we can choose $P=(0,\dots,1,\dots,0)$ with the 1 in position $i$ and $P'=(0,\dots,0)$. We find that $Z_{ii}$ is in $\gg(E)$.

Then, we have that for all $P\in \{0,1\}^N$
\[\left[X_{P}(E),\begin{pmatrix}0&I_N\\
                I_N&0
               \end{pmatrix}\right]=
              -2E\begin{pmatrix}
                    I_N&0\\
                    0&-I_N
                   \end{pmatrix}.
               \]
Therefore, since $E\neq 0$,  $\begin{pmatrix}     I_N&0\\
                    0&-I_N
                   \end{pmatrix}\in\gg(E)$.
                   We can then see that for all $i$ in $\{1,\dots, N\}$
                   \[\left[X_{ii},\begin{pmatrix}
                    I_N&0\\
                    0&-I_N
                   \end{pmatrix}\right]=-2Y_{ii},\]
                   and thus all the $Y_{ii}$ are in $\gg(E)$.      
We have thus that
\begin{equation*}
    \{\left(\begin{smallmatrix} 0&D_1\\
D_2&0
\end{smallmatrix}\right), D_1,\ D_2\text{ diagonal}\}\subset \gg(E).
\end{equation*}
Therefore, 
\begin{equation*}
 X_P(E)-  \left( \begin{smallmatrix}
0 & E + \sum_{P_i=1} E_{ii} \\
-E+  \sum_{P_i=1} E_{ii} & 0
\end{smallmatrix} \right)= \left(\begin{smallmatrix}
0&\Delta\\                                                                    
\Delta&0                                                      \end{smallmatrix}\right)\in\gg(E).
\end{equation*}
 We  compute, for $i$ in $\{1,\dots,N\}$,
\[\left[ \begin{pmatrix}
  0&\Delta\\\Delta&0                                     \end{pmatrix}, Y_{ii}\right]=-2\begin{pmatrix}
                B_i&0\\
                0& -B_i
               \end{pmatrix},\]
where $B_i:=\Delta E_{ii}+E_{ii}\Delta=E_{i,i-1}+E_{i,i+1}+E_{i-1,i}+E_{i+1,i}$ with the convention that $E_{ij}=0$ if $i$ or $j$ is not in $[1,N]$.
We remark that  \begin{equation*}
    \text{Span}(\{B_i, i=1,\dots,N-1\})=\text{Span}(\{E_{i,i+1}+E_{i+1,i}, i=1,\dots,N-1\}),
\end{equation*}
which is the space of symmetric tridiagonal matrices with zeros on the diagonal.
 
 Now, for any $i$ in $\{1,\dots,N-1\}$,
 \[\left[\begin{pmatrix}
          E_{i,i+1}+E_{i+1,i}&0\\
          0&-(E_{i,i+1}+E_{i+1,i})
         \end{pmatrix},\begin{pmatrix}0&I_N\\
         I_N&0\end{pmatrix}\right]=2Y_{i,i+1}.
         \]
As a consequence,  all the $Y_{ij}$'s with $|i-j|=1$ are in $\gg(E)$.
Last, 
\[\left[\begin{pmatrix}
         E_{i,i+1}+E_{i+1,i}&0\\
          0&-(E_{i,i+1}+E_{i+1,i})
         \end{pmatrix},\begin{pmatrix}0&I_N\\
         -I_N&0\end{pmatrix}\right]=2Z_{i,i+1}.
         \]
We have that all the $Z_{ij}$'s with $|i-j|=1$ are in $\gg(E)$.
With Lemma~\ref{lem:struc_spN}, this concludes the proof that $\gg(E)=\spN$.

\subsubsection*{Case 3: deterministic potential on $\sigma_0$, random potential on $\sigma_3$.}
Recall that, according to Table~\ref{tab:cas}, this covers as well the case with the deterministic potential on $\sigma_0$ and the random potential on $\sigma_1$. We consider the family of matrices $\{X_P(E)\}_{P\in\{0,1\}^N}$, as defined in~\eqref{def_Xn}, with $\alpha_0=\beta_3=1$ and all the other coefficients equal to 0.
In other words,  \[X_P(E)=\left( \begin{smallmatrix}
0 & E-\Delta + \sum_{P_i=1} E_{ii} \\
-E+\Delta +  \sum_{P_i=1} E_{ii} & 0
\end{smallmatrix} \right).\]

We  prove that all the $Z_{ij}$'s and $Y_{ij}$'s, for $|i-j|\leq1$, are in $\gg(E)$.
First, we see that
\begin{equation}
    X_{P}(E)-X_{P'}(E)=\left(\begin{smallmatrix}                                                                        0&\sum_{P_i=1} E_{ii}-\sum_{P'_i=1} E_{ii}\\                                                                       \sum_{P_i=1} E_{ii}-\sum_{P'_i=1} E_{ii}&0                                                                                     \end{smallmatrix}\right)
\end{equation}
so, for all $i$,  $Z_{ii}\in\gg(E)$.

Then,  \[\left( \begin{smallmatrix}
0 & E-\Delta \\
-E+\Delta  & 0
\end{smallmatrix} \right)\in\gg(E).\]
Taking the bracket, we find that for all $i$ in $\{1,\dots,N\}$,
\begin{equation}\label{commut:DDelta}\left[\left( \begin{smallmatrix}
0 & E-\Delta \\
-E+\Delta  & 0
\end{smallmatrix}\right),Z_{ii}\right]=\left( \begin{smallmatrix}
2EE_{ii}-\{E_{ii},\Delta\} &0 \\0&
-2EE_{ii}+\{E_{ii},\Delta\}  
\end{smallmatrix}\right) \in\gg(E).\end{equation}

Recall that we have defined the matrix $K=\sum_{i=1}^N(-1)^{i+1}E_{ii}$. We find that  $\{K,\Delta\}=0$. As a consequence, since $E\neq0$, we find that
\begin{equation}
\sum_{i=1}^N(-1)^{i+1}\left[\left( \begin{smallmatrix}
0 & E-\Delta \\
-E+\Delta  & 0
\end{smallmatrix}\right),Z_{ii}\right]=2E
\left(\begin{smallmatrix}
K &0 \\0&
-K  
\end{smallmatrix}\right) \in\gg(E).\end{equation}
Taking another bracket, we find that, for all $i\in\{1,\dots,N\}$,
\begin{equation}
\left[\left( \begin{smallmatrix}
K &0 \\0&
-K 
\end{smallmatrix}\right),\left( \begin{smallmatrix}
0 & E_{ii} \\
E_{ii} & 0
\end{smallmatrix}\right)\right]=\left( \begin{smallmatrix}
0 &2(-1)^{i+1}E_{ii} \\-2(-1)^{i+1}E_{ii}&
0  
\end{smallmatrix}\right)=(-1)^{i+1}Y_{ii} \in\gg(E).
\end{equation}

On the other hand,  we find that
\begin{equation}
\Theta:=\sum_{i=1}^N\left[\left( \begin{smallmatrix}
0 & E-\Delta \\
-E+\Delta  & 0
\end{smallmatrix}\right),Z_{ii}\right]=\left( \begin{smallmatrix}
E-\Delta &0 \\0&
-E+\Delta 
\end{smallmatrix}\right) \in\gg(E).
\end{equation}
For $i$ between 2 and $N-1$,
\begin{equation}
\frac12[\Theta,Z_{ii}]=\left( \begin{smallmatrix}
0 &2EE_{ii}-\{E_{ii},\Delta\} \\-2EE_{ii}+\{E_{ii},\Delta\}&
0
\end{smallmatrix}\right)=EY_{ii}+Y_{i,i+1}+Y_{i,i-1}\end{equation} and
\begin{equation}
\frac12[\Theta,Z_{11}]=\left( \begin{smallmatrix}
0 &2EE_{11}-\{E_{11},\Delta\} \\-2EE_{11}+\{E_{11},\Delta\}&
0
\end{smallmatrix}\right)=EY_{11}+Y_{12}.\end{equation}
This, together with the fact that all $Y_{ii}$'s are in $\gg(E)$, implies that, for all $i=1,\dots,N-1$, $Y_{i,i+1}\in\gg(E)$. Similarly, we prove, computing the commutators $[\Theta,Y_{ii}]$ that all $Z_{i,i+1}$ are in $\gg(E)$. As a consequence, by Lemma~\ref{lem:struc_spN}, $\gg(E)=\sp_N(\R)$.

\subsubsection*{Case 4: deterministic potential on $\sigma_1$, random potential on $\sigma_3$}

Recall that, according to Table~\ref{tab:cas}, this covers as well the case where the deterministic potential is on $\sigma_3$ and the random potential is on $\sigma_1$. We consider the family of matrices $\{X_{P}(E)\}_{P\in \{0,1\}^N}$, as defined in~\eqref{def_Xn}, with $\alpha_1=\beta_3=1$ and all other coefficients equal to 0.
As a consequence,  \[X_{P}(E)=\left( \begin{smallmatrix}-\Delta & E+ \sum_{P_i=1} E_{ii} \\
-E+   \sum_{P_i=1} E_{ii}&  \Delta 
\end{smallmatrix} \right)                                       .\]
As in the previous cases, we will use Lemma~\ref{lem:struc_spN} and only prove that all the $Z_{ij}$'s and $Y_{ij}$'s, for $|i-j|\leq1$, are in $\gg(E)$.
First, we see that for all $P,P'$
\begin{equation}
    X_P(E)-X_{P'}(E)=\sum_{P_i=1}Z_{ii}-\sum_{P'_i=1}Z_{ii}\in\gg(E)
\end{equation}
so, for all $i$, $Z_{ii}\in\gg(E)$.

Then,  \[X_P(E)-\sum_{P_i=1}Z_{ii}=\left( \begin{smallmatrix}
-\Delta & E \\
-E&\Delta  
\end{smallmatrix} \right)\in\gg(E).\]
Taking the bracket, we find that for all $i$ in $\{1,\dots,N\}$,
\begin{equation}\label{commut:DDeltabis}\left[\left( \begin{smallmatrix}
-\Delta & E \\
-E&-\Delta  
\end{smallmatrix}\right),Z_{ii}\right]=\left( \begin{smallmatrix}
-2EE_{ii}&-\{E_{ii},\Delta\}\\\{E_{ii},\Delta\}  &2EE_{ii}
\end{smallmatrix}\right) \in\gg(E).\end{equation}

Similarly to the previous case, we can write
\begin{equation}
\sum_{i=1}^N(-1)^{i+1}\left[\left( \begin{smallmatrix}
0 & E-\Delta \\
-E+\Delta  & 0
\end{smallmatrix}\right),X_{ii}\right]=\left(\begin{smallmatrix}
2EK &0 \\0&
-2EK  
\end{smallmatrix}\right) \in\gg(E).\end{equation}
Taking another bracket, we find that, for all $i$,
\begin{equation}
\left[\left( \begin{smallmatrix}
K &0 \\0&
-K 
\end{smallmatrix}\right),Z_{ii}\right]=(-1)^{i+1}Y_{ii} \in\gg(E).
\end{equation}

This implies that $\Theta':=\left( \begin{smallmatrix}
-\Delta&0 \\0&
\Delta  
\end{smallmatrix}\right)\in\gg(E).$
We can conclude as in the previous case.

\subsubsection*{Case 5: deterministic potential on $\sigma_3$, random potential on $\sigma_0$.}
Recall that, according to Table~\ref{tab:cas}, this covers as well the case with the deterministic potential on $\sigma_1$ and the random potential on $\sigma_0$. We consider the family of matrices $\{X_P(E)\}_{P\in  \{0,1\}^N}$, as defined in~\eqref{def_Xn}, with $\alpha_3=\beta_0=1$ and all other coefficients equal to 0.
We have
\begin{equation}
 X_{P}(E)=\left( \begin{smallmatrix}
0 & E+\Delta - \sum_{P_i=1} E_{ii} \\
-E+\Delta +  \sum_{P_i=1} E_{ii} & 0
\end{smallmatrix} \right)
\end{equation}
Here, the Lie algebra $\gg(E)$  will not be $\spN$ but its subalgebra $\spONN$. We have the following result, similar to Lemma~\ref{lem:struc_spN}.
\begin{lemma}\label{lem:struc_spON}
We adopt the same notation as in Lemma~\ref{lem:struc_spN}.
If some subalgebra $\gg$ of $\spONN$ contains all the  $Y_{ii}$'s as well as the $X_{ij}$'s for $|i-j|= 1$ , then $\gg=\spONN$.
  \end{lemma}

\begin{proof}
We first remark that all the $Y_{ii}$'s as well as the $X_{ij}$'s for $|i-j|= 1$ are in $\spONN$. To prove the other inclusion, let $\gg$ be a subalgebra of $\spONN$ which contains all these matrices.
Recall that a basis of $\{A\in \mathcal{M}_{N}(\R)\ |\ ^sA=A \}$ is given by $\{ E_{ij} + (-1)^{i-j+1}E_{ji}\}_{1\leq i\leq j\leq N}$ and one of $\{B\in \mathcal{M}_{N}(\R)\ |\ ^tB=B \}$ is given by $\{ E_{ij} + E_{ji}\}_{1\leq i\leq j\leq N}$. We introduce, for every $i,j\in \{1,\ldots N\}$, the matrices
$$V_{ij}= \left(\begin{smallmatrix}
     E_{ij}+(-1)^{i-j+1} E_{ji}&0\\0& (-1)^{i-j} E_{ij}-E_{ji}
    \end{smallmatrix}\right)\ \mbox{ and } \ W_{ij}= \left(\begin{smallmatrix}
    0 & E_{ij}+E_{ji}\\ (-1)^{i-j+1} (E_{ij}+E_{ji}) & 0    \end{smallmatrix}\right).$$
By \eqref{eq_expr_sponn}, $\{V_{ij},W_{ij} \}_{1\leq i\leq j \leq N}$ is a basis of $\spONN$. Hence, to prove Lemma \ref{lem:struc_spON}, it suffices to prove that $\gg$ contains all the $V_{ij}$'s and $W_{ij}$'s for $1\leq i\leq j \leq N$. A direct computation shows that for every $i,j,k,r\in \{1,\ldots,N\}$, 
\begin{align}
[V_{ij},W_{kr}] & = \delta_{jk} \left( \begin{smallmatrix} \mathrm{I}_{\mathrm{N}} & 0\\ 0 & (-1)^{k-j} \mathrm{I}_{\mathrm{N}} \end{smallmatrix} \right) W_{ir} + \delta_{ik} (-1)^{i-j+1}\left( \begin{smallmatrix} \mathrm{I}_{\mathrm{N}} & 0\\ 0 & (-1)^{i+k} \mathrm{I}_{\mathrm{N}} \end{smallmatrix} \right) W_{jr} \nonumber \\
 & \qquad \qquad + \delta_{ir} (-1)^{i-j+1}\left( \begin{smallmatrix} \mathrm{I}_{\mathrm{N}} & 0\\ 0  & (-1)^{i-r} \mathrm{I}_{\mathrm{N}} \end{smallmatrix} \right) W_{jk} + \delta_{jr} \left( \begin{smallmatrix} \mathrm{I}_{\mathrm{N}} & 0\\ 0 & (-1)^{-j-r} \mathrm{I}_{\mathrm{N}} \end{smallmatrix} \right) W_{ik} \nonumber \\
& = (-1)^{i-j+1} (\delta_{ik} W_{jr} + \delta_{ir} W_{jk}) + \delta_{jk} W_{ir} + \delta_{jr} W_{ik} \label{eq_brackets_lemma_spon}
\end{align}
and if $(-)_{ijkr} = (-1)^{i-j+k-r+1}$,
\begin{align}
 [W_{ij},W_{kr}] &  =(-1)^{k-r+1}\left[ \delta_{jk}  \left( \begin{smallmatrix} E_{ir}+(-)_{ijkr} E_{ri} & 0\\ 0 & -((-)_{ijkr} E_{ir} + E_{ri})  \end{smallmatrix} \right) + \delta_{jr} \left( \begin{smallmatrix} E_{ik}+(-)_{ijkr} E_{ki} & 0\\ 0 & -((-)_{ijkr} E_{ik} + E_{ki})  \end{smallmatrix} \right) \right. \nonumber \\
 &\qquad \left.  + \delta_{ik} \left( \begin{smallmatrix} E_{jr}+(-)_{ijkr} E_{rj} & 0\\ 0 & -((-)_{ijkr} E_{jr} + E_{rj})  \end{smallmatrix} \right) + \delta_{ir} \left( \begin{smallmatrix} E_{jk}+(-)_{ijkr} E_{kj} & 0\\ 0 & -((-)_{ijkr} E_{jk} + E_{kj})  \end{smallmatrix} \right) \right] \nonumber \\
& = (-1)^{k-r+1} (\delta_{jk} V_{ir} + \delta_{jr} V_{ik} + \delta_{ik} V_{jr} + \delta_{ir} V_{jk} )\label{eq_brackets_lemma_spon2}
\end{align}

where $\delta_{ij}$ is the Kronecker's symbol : 
$$\delta_{ij}=\left\lbrace \begin{array}{ccl}
1 & \mathrm{if} & i=j \\
0 & \mathrm{if} & i\neq j.
\end{array} \right.$$
The hypothesis made on $\gg$ implies that for every $i\in \{1,\ldots,N\}$, $W_{ii}\in \gg$ and for $|i-j|=1$, $W_{ij}\in \gg$. Let $i\in \{1,\ldots,N\}$. Then, using \eqref{eq_brackets_lemma_spon2}, $[W_{ii},W_{i,i+1}]=2V_{i,i+1}$ and $V_{i,i+1}\in \gg$. This proves that $\gg$ contains all the $V_{ij}$ and $W_{ij}$ for $|i-j|=1$. Using \eqref{eq_brackets_lemma_spon}, $[V_{i,i+1},W_{i+1,i+2}] = W_{i,i+2}$ and $W_{i,i+2} \in \gg$. Using \eqref{eq_brackets_lemma_spon2} this implies that $[W_{ii},W_{i,i+2}]=-2V_{i,i+2}$ and $V_{i,i+2} \in \gg$. Hence, $\gg$ contains all the $V_{ij}$ and $W_{ij}$ for $|i-j|=2$. By induction, using \eqref{eq_brackets_lemma_spon} and \eqref{eq_brackets_lemma_spon2}, we prove that $\gg$ contains  all the $V_{ij}$ and $W_{ij}$ for $i\leq j$, hence $\spONN \subset \gg$ and $\gg=\spONN$.
\end{proof}

To compute $\gg(E)$, we begin as usual by taking $X_P(E)-X_{P'}(E)$ to find that all the $Y_{ii}$'s are in $\gg(E)$.
Then, \[\begin{pmatrix}0&\Delta\\\Delta&0\end{pmatrix}\in\gg(E).\]
This implies that
\[\left[\begin{pmatrix}0&-\Delta\\\Delta&0\end{pmatrix},\begin{pmatrix}
0&I_N\\-I_N&0\end{pmatrix}\right]=-2\begin{pmatrix}\Delta&0\\0&-\Delta\end{pmatrix}\in\gg(E).\]
Taking commutators between this last matrix and the $Y_{ii}$'s, we find that all the $Z_{ij}$'s with $|i-j|=1$ are in $\gg(E)$. \qed

\subsection{$p$-contractivity and $L_p$-strong irreducibility}\label{sec:cas5}
As we already explained,  $\SpN$ is $p$-contracting and $J$-$L_p$-strongly irreducible for all $p\in\{1,\dots,N\}$.
We prove here the corresponding result for $\SpONN$.
We prove a similar result for $\SpONN$.
\begin{prop}\label{prop:G5CSI}
    The group $\SpONN$ is $2p$-contracting and $(J,S)$-$L_{2p}$-strongly irreducible for all $p\in \{1,\ldots, N\}$.
\end{prop}
\begin{proof}
    According to \cite[Proposition~A.IV.2.1]{BL}, a subset of $\GL_{2N}(\R)$ is $2p$-contracting if there exists in it a sequence $(M_n)$ such that, if we denote by $s_1(M_n)\geq\dots\geq s_{2N}(M_n)$ the singular values of $M_n$, then $\lim_{n\to\infty}s_{2p+1}(M_n)/s_{2p}(M_n)=0$. But, if we take a sequence $t_1>\dots>t_d>0$ and construct the matrix $R$ as in Proposition~\ref{prop:decompG}~(iv), then for all $n\in \N$ the singular values of $R^n$ are, for $p\leq d$, $s_{2p-1}(M_n)=s_{2p}(M_n)=e^{nt_p}$ so the sequence satisfies the criterion: $\SpONN$ is $2p$-contracting.

Let us now prove that $\SpONN$ is $L_{2p}$-strongly irreducible. Since $\SpONN$ is connected (see Proposition~\ref{prop:decompG}), according to \cite[Exercise~A.IV.2.9]{BL}, we only have to prove that there exists no proper subspace $V$ of $L_{2p}$ such that $(\Lambda^{2p}M)(V)\subset V$ for all $M$ in $\SpONN$. Assume that such a $V$ exists and let us consider a matrix $R$ as in Proposition~\ref{prop:decompG}~(iv), with $t_1>\dots>t_d$. We have in particular that $(\Lambda^{2p}R)(V)\subset V$. For $f_1$ defined in~\eqref{def:f1}, let us write the unique decomposition $f_1=v+v^{\bot}$ with $v\in V$ and $v^{\bot}\in V^{\bot}$. Then, 
\begin{equation}\label{eq:propG5SCI}
(\Lambda^{2p}R)f_1 =e^{2t_1}\cdots e^{2t_{p}} f_1 =(\Lambda^{2p}R)v + (\Lambda^{2p}R)v^{\bot}.    
\end{equation}
Since $V$ is stable by $\Lambda^{2p}R$ and $V^{\bot}$ is stable by $(\Lambda^{2p}R)^*=\Lambda^{2p}R$, $(\Lambda^{2p}R)v\in V$ and $(\Lambda^{2p}R)v^{\bot}\in V^{\bot}$. By orthogonality, $(\Lambda^{2p}R)v=e^{2t_1}\cdots e^{2t_{p}} v$. If $v=0$ then $f_1 \in V^{\bot}$. If $v\neq 0$, $v$ is an eigenvector of $\Lambda^{2p}R$ associated with the eigenvalue $e^{2t_1}\cdots e^{2t_{p}}$. But,
by construction of $R$, its eigenspace associated with the eigenvalue $e^{2t_1}\cdots e^{2t_{p}}$ is $\mathrm{Span}(f_1)$. Hence, if  $v$ is an eigenvector of $\Lambda^{2p}R^n$ associated to the eigenvalue $e^{2nt_1}\cdots e^{2nt_{p}}$, then $v\in \mathrm{Span}(f_1)$ and $f_1\in V$. Thus, we proved that either $f_1\in V$ or $f_1\in V^{\bot}$.

If $f_1\in V$, then by definition $V=L_{2p}$. If $f$ is orthogonal to $V$, then, for all $v\in V$ and $M\in \SpONN$, $\langle \Lambda^{2p}Mf_1,v\rangle=\langle f_1,\Lambda^{2p}(^tM)v\rangle=0$ since $^tM$ is in $\SpONN$ as well. Consequently, $V=L_{2p}^\perp$, which contradicts the hypothesis that $V$ is proper. 
\end{proof}




\subsection{Proof of Theorem \ref{thm:localization:exemple}}\label{subsec:proof-loc-exemple}

The proof of Theorem \ref{thm:localization:exemple} comes from the application of one of the localization criteria or the delocalization criterion (\emph{i.e.} Theorem \ref{thm:loc_criterion1}, Theorem~\ref{thm:loc_criterion2} or Theorem \ref{thm:deloc_crit}) to $\{ \DO \}_{\omega\in \Omega}$ for the particular choice of $V=\Delta$, the tridiagonal matrix with $0$ on the diagonal and $1$ on the upper and lower diagonal.

\begin{proof}[Proof of Theorem \ref{thm:localization:exemple}]
Once we obtained Theorem \ref{thm:localization:exemple} for $V=\Delta$ we apply the genericity argument developed in \cite{Bou13} in the following way. We fix $G \in \{ \SpN, \SpONN\}$  and $\mathfrak{g}\in \{ \spN, \spONN \}$ the corresponding Lie algebra. Both $\SpN$ and $\SpONN$ are algebraic groups. The use of Theorem \ref{thm_breuillard} to obtain the separability of Lyapunov exponents leads us to prove an algebraic property on a Lie algebra generated by a finite number of matrices. Hence, the $n$-tuples of elements in $G$ that do not generate a dense subgroup are contained in a closed analytic subvariety which implies that we can perturb the interaction potential $\Delta$ into a potential $V$ while keeping the property that the Furstenberg group is equal to $G$ for any energy in an interval $I(N,\ell,V)$ for $\ell \in (0,\ell_C(N,V))$. The transfer matrices are written in exponential form,
$$T_{\omO}(E)=\exp(\ell X_{\omO}(E,V)),$$ 
where the $X_{\omO}(E,V)$ are define in \eqref{def_Xn} and where we explicit the dependency in $V$ of these matrices since $E$ and $V$ are both important variables in the genericity argument. We denote the family $\{X_{\omO}(E,V)\}_{\omO\in \{0,1\}^N}$ by ${X_1(E,V),\ldots, X_{2^N}(E,V)}$. For $k\in \N^*$, let
\begin{equation}\label{eq_def_Vk_AB_N_generic}
\mathcal{V}_k=\left\{ (X_1,\ldots,X_k)\in \spN^k \ |\ (X_1,\ldots,X_k)\ \mathrm{does\ not\ generate}\ \mathfrak{g} \right\}.
\end{equation}
Since generating the algebra $\mathfrak{g}$ is an algebraic condition of the type non-vanishing of a finite family of determinants (finite because, for all $m\in \N^*$, the ring $\R[T_1,\ldots,T_m]$ is Noetherian), there exist $Q_1,\ldots,Q_{r_k}\in \R[\mathfrak{g}^k]$ such that:
\begin{equation}\label{eq_prop_Vk__AB_N_generic}
\mathcal{V}_k=\left\{ (X_1,\ldots,X_k)\in \mathfrak{g}^k \ |\ Q_1(X_1,\ldots,X_k)=0,\ldots, Q_{r_k}(X_1,\ldots,X_k)=0 \right\}.
\end{equation} 
Here, we identify $\R[\mathfrak{g}^k]\simeq \R[T_1,\ldots,T_{k\times \dim \mathfrak{g}}]$ where $\dim \mathfrak{g}=2N^2 +N$ for $\mathfrak{g}=\spN$ and $\dim \mathfrak{g}=N^2$ for $\mathfrak{g}=\spONN$. Let $E\in \R$ and,
\begin{equation}\label{eq_def_VE_AB_N_generic}
\mathcal{V}(E)=\left\{ V\in \SN \ |\ { X_1(E,V),\ldots,X_{2^N}(E,V)}\ \mathrm{does\ not\ generate}\ \mathfrak{g} \right\}.
\end{equation}

\noindent We show that $\mathrm{Leb}_{\frac{N(N+1)}{2}}(\mathcal{V}(E))=0$. Indeed, let 
\begin{equation}
f_E\ :\ \begin{array}{ccl}
\SN & \to & \mathfrak{g}^{2^N} \\
V & \mapsto & (X_1(E,V),\ldots,X_{2^N}(E,V))
\end{array}.
\end{equation}

\noindent Then, $f_E$ is polynomial in the $\frac{N(N+1)}{2}$ coefficients defining $V$, and we have:
\begin{equation}\label{eq_prop_VE_AB_N_generic}
V\in \mathcal{V}(E) \ \Leftrightarrow \ (Q_1\circ f_E)(V)=0,\ldots,(Q_{r_{2^N}}\circ f_E)(V)=0,
\end{equation}
where each $Q_i \circ f_E$ is polynomial in the $\frac{N(N+1)}{2}$ coefficients defining $V$. But, we have shown in Section \ref{sec:trans_mat} that in Cases 2, 3, 4 and 5 for $N$ even, $\Delta \notin \mathcal{V}(E)$. Therefore, there exists $i_0\in \{ 1,\ldots, r_{2^N}\}$ such that $(Q_{i_0}\circ f_E)(\Delta)\neq 0$, and since the function $Q_{i_0}\circ f_E$ is polynomial and not identically zero,
\begin{equation}\label{eq_prop_Leb_VE_AB_N_generic}
\mathrm{Leb}_{\frac{N(N+1)}{2}}\left( \{ V\in \SN \ |\ (Q_i \circ f_E)(V)=0) \} \right)=0,
\end{equation}
and, by inclusion,
\begin{equation}\label{eq_prop_Leb_VE_2}
\forall E\in \R,\  \mathrm{Leb}_{\frac{N(N+1)}{2}}(\mathcal{V}(E))=0.
\end{equation}
Finally, let $\mathcal{V}(\mathfrak{g})=\cap_{E\in \R} \mathcal{V}(E)$. Then $\mathcal{V}(\mathfrak{g})$ has Lebesgue measure zero, and if $V\notin \mathcal{V}(\mathfrak{g})$, there exists $E_0\in \R$ such that the family $\{ X_1(E_0,V),\ldots,X_{2^N}(E_0,V)\}$ generates $\mathfrak{g}$. Therefore, there exists $i_0\in \{ 1,\ldots, r_{2^N}\}$ such that $(Q_{i_0} \circ f)(E_0,V)\neq 0$, where:
\begin{equation}\label{eq_def_f_AB_N_generic}
f\ :\ \begin{array}{ccl}
\R \times \SN & \to & \mathfrak{g}^{2^N} \\
(E,V) & \mapsto & (X_1(E,V),\ldots, X_{2^N}(E,V))
\end{array}.
\end{equation}
Now, for $V$ fixed, $E\mapsto (Q_{i_0}\circ f_E)(E,V)$ is polynomial and not identically zero, so it has only a finite set $\mathcal{S}_{\mathrm{V}}$ of zeros, and for all $E\in \R\setminus \mathcal{S}_{\mathrm{V}}$, $\{ X_1(E,V),\ldots,X_{2^N}(E,V)\}$ does not generate $\mathfrak{g}$.

Thus we have obtained that $\mathcal{V}(\mathfrak{g})$ has Lebesgue measure zero, and if $V\notin \mathcal{V}(\mathfrak{g})$, there exists $\mathcal{S}_{\mathrm{V}} \subset \R$ finite such that for all $E\in \R\setminus \mathcal{S}_{\mathrm{V}}$, $\{ X_1(E,V),\ldots,X_{2^N}(E,V)\}$ generates $\mathfrak{g}$.

From there, we finish the proof of Theorem \ref{thm:localization:exemple}. In Cases 2, 3, 4, we fix $V\in \SN \setminus \mathcal{V}(\spN)$ and apply Theorem \ref{thm_breuillard}, using the real number $\ell_C(N,V)$ and the interval $I(N,V,\ell)$ given by Lemma \ref{lem:normX}. Then we get that the hypotheses of Theorem \ref{thm:loc_criterion1} are satisfied for this $V$ on $I(N,V,\ell)\setminus \mathcal{S}_{\mathrm{V}}$ since $G(E)=\SpN$ on this interval. In Case 5 we fix $V\in \SN \setminus \mathcal{V}(\spONN)$ and we do the same as in Cases 2, 3, 4  to get that the hypotheses of Theorem \ref{thm:loc_criterion2} are satisfied since $G(E)=\SpONN$ on the constructed interval.
\end{proof}

\begin{rem}
It is the algebraic nature of the objects involved that allows us to prove a generic result in $V$ and the finiteness of the set of critical energies. We can summarize the ideas used simply by recalling that the set of zeros of a non-zero single-variable polynomial is finite and that more generally, the set of zeros of a non-zero polynomial in several variables is of Lebesgue measure zero.
\end{rem}

\appendix
\section{The group $\SpONN$}\label{app:G5}
In order to study the Lyapunov exponents associated with a sequence of i.i.d. matrices of $\SpONN$, as defined by~\eqref{def:SpONN}, we state here the most important properties of this group.

\begin{prop}\label{prop:decompG}
Let $M\in \SpONN$.
 \begin{enumerate}
  \item $^tM\in \SpONN$.
  \item If $Mv=\lambda v$ for $v\in\R^{2N}$ and $\lambda\in\R$, then $^tMJv=\lambda^{-1}Jv$, $^tMSv=\lambda^{-1}Sv$ and $MJSv=\lambda JSv$.
  \item For all $v\in\R^{2N}$, the vectors $v$ and $JSv$ are orthogonal.
  \item For $t\in\R$, we denote $B_t:=\begin{pmatrix} \cosh t&\sinh t\\ \sinh t&\cosh t\end{pmatrix}$. If $N$ is even, we denote $N=:2d$ and, if $N$ is odd, $N=:2d+1$. There exists $U\in \SO_{2N}(\R)\cap \SpONN$ and real numbers $t_1\geq\dots\geq t_d\geq0$ such that, if we denote  in the even case\[R:=\begin{pmatrix} \diag(B_{t_1},\dots,B_{t_d})&0\\0&\diag(B_{-t_1},\dots,B_{-t_d})\end{pmatrix},\]
  in the odd case\[R:=\begin{pmatrix} \diag(B_{t_1},\dots,B_{t_d},1)&0\\0&\diag(B_{-t_1},\dots,B_{-t_d},1)\end{pmatrix},\]
  then
  \begin{equation}
U^tMM^tU=R.
  \end{equation}
 \item There exists $K$, $U\in \SO_{2N}(\R)\cap \SpONN$ and a matrix $R$ similar to the one introduced above such that
 \begin{equation}
  M=KRU.
 \end{equation}

\item The group $\SpONN$ is pathwise connected.

\item The Lie algebra of the Lie group $\SpONN$ is denoted by $\spONN$ and is given by  
$$\spONN:= \left\{\begin{pmatrix}
     A&B\\^sB&-^tA
    \end{pmatrix}
\text{ with } ^sA=A, ^tB=B\right\}$$
where the operator $^s$ is defined on $\mathcal{M}_{N}(\R)$ by
$(^sM)_{ij}=(-1)^{i-j+1}M_{ji}$. 

\item The group $\SpONN$ is semisimple.
 \end{enumerate}

\end{prop}
\begin{proof}
 \begin{enumerate}
  \item If we take the inverse of the relations defining $\SpONN$, we find, since $S^{-1}=S$ and $J^{-1}=-J$,
  \begin{equation*}
   M^{-1}S (^tM)^{-1}=S\text{ and } M^{-1}J (^tM)^{-1}=J,
  \end{equation*}
so $(^tM)^{-1}\in \SpONN$, which implies that it is the case for $^tM$ as well, since $\SpONN$ is a group.

\item If $Mv=\lambda v$, then $\lambda^tMJv= ^tMJMv=Jv$. The same argument is true if we replace $J$ by $S$. Combining these two results, we find the one with $JS$.

\item Since $JS=\begin{pmatrix} 0&-K\\K&0\end{pmatrix}$, for $v=\begin {pmatrix}v_1\\v_2\end{pmatrix}$, we have \[v\cdot JSv=-v_1\cdot Kv_2+v_2\cdot Kv_1=0\] since the matrix $K$ is symmetric.

\item 
 We first prove that $^tMM$ has an orthonormal basis of eigenvectors which can be written $$(v_1,\dots,v_d,Jv_1,\dots, Jv_d,Sv_1,\dots,Sv_d,JSv_1,\dots,JSv_d)$$ when $N=2d$ and  $$(v_1,\dots,v_d,v_{d+1},Jv_1,\dots, Jv_d,Jv_{d+1}Sv_1,\dots,Sv_d,JSv_1,\dots,JSv_d)$$ when $N=2d+1$. To this purpose, we split the space $\R^{2N}$  into 2 subspaces: $V_1$ is the eigenspace of $^tMM$ associated with the eigenvalue 1 and $V_0=V_1^\perp$.
 
 Let us first consider $V_0$: we prove by induction that its dimension is a multiple of 4 and that it has a basis of the desired form. If the dimension is 0, there is nothing to prove. Otherwise, we take an eigenvector  $v_1$ of  $^tMM$, associated with its largest eigenvalue $e^{t_1}>1$ which has norm 1. We see by Point (2) that $JSv$ is an eigenvector associated with  the same eigenvalue and, by Point (3), it is orthogonal to $v_1$. It has norm 1 as well. Since $^tMM$ is symmetric, $Jv_1$ and $Sv_1$ are by Point (2) eigenvectors of it as well,   associated with the eigenvalue  $e^{-t_1}\neq e^{t_1}$. They both have norm 1, are orthogonal to each other (since $Jv_1=JSSv_1$), and orthogonal to $v_1$ and $JSv_1$ since the eigenspaces associated with different eigenvalues are orthogonal. Then, the space spanned by the 4 vectors $v_1$, $Jv_1$, $Sv_1$ and $JSv_1$ has dimension 4 and has an orthonormal basis of the desired form. We can now consider the orthogonal of the space spanned by these 4 eigenvectors to apply the induction hypothesis.
 
 We now consider the space $V_1$. Note that this space is stable under $S$. Therefore, $V_1=V_1^+\oplus V_1^-$, where the spaces $V_1^\pm:=\{w\in V_1, Sw=\pm w\}$, are orthogonal. These 2 spaces are moreover isomorphic since, if $w\in V_1^+$, then $Jw\in V_1^-$ and \emph{vice versa}. Let us assume that dim$(V_1)\geq4$, which implies that dim $(V_1^+)\geq2$. We can consider two orthogonal vectors $w_1$, $w_2\in V_1^+$, both with norm 1. If we define $v_{r+1}=(w_1+Jw_2)/\sqrt{2}$, we can see that the vectors $v_{r+1}$, $Jv_{r+1}=(Jw_1-w_2)/\sqrt{2}$, $Sv_{r+1}=(w_1-Jw_2)/\sqrt{2}$ and $JSv_{r+1}=(Jw_1+w_2)/\sqrt{2}$ are orthonormal. We can complete the basis by applying the same process to the orthogonal of the space spanned by these 4 vectors in $V_1$. If  $N$ is even, the dimension of $V_1$ is a multiple of 4 so we can construct the whole basis in this way. If $N$ is odd, we are left with a 2-dimensional subspace of $V_1$, of which we can take an orthonormal basis of the form $(v_{d+1},Jv_{d+1})$. 

We now construct the matrix $U$ in the following way. 
    For $i=1,\dots,d$, we define $Uv_i=1/\sqrt{2}(e_{2i}+e_{2i-1})$, $UJv_i=1/\sqrt{2}(e_{2i+N}+e_{2i-1+N})$, $USv_i=1/\sqrt{2}(e_{2i}-e_{2i-1})$ and $UJSv_i=1/\sqrt{2}(e_{2i+N}-e_{2i-1+N})$. If $N$ is odd, we define too $Uv_{d+1}=e_N$ and $UJv_{d+1}=e_{2N}$. We easily see that such a $U$ is unitary. Moreover, by construction, $U^tMM^tUe_{2i-1}=\cosh t_i e_{2i-1}+\sinh t_i e_{2i}$, $U^tMM^tUe_{2i}=\cosh t_i e_{2i}+\sinh t_i e_{2i-1}$, $U^tMM^tUe_{2i-1+N}=\cosh t_i e_{2i-1+N}-\sinh t_i e_{2i+N}$ and $U^tMM^tUe_{2i+N}=\cosh t_i e_{2i+N}-\sinh t_i e_{2i-1+N}$, so $U^tMM^tU=R$. We can explicitly compute that $UJ^tU e_i=e_{N+i}$ and $UJ^tUe_{N+i} =-e_i$ so $UJ^tU=J$. Similarly, we prove that $US^tU=S$ so $U\in \SpONN$.
    
    Since $U$ is a symplectic matrix,   its determinant is 1, so $U\in SO_{2N}(\R)$.
    
    \item Let us use the previous point to write that there exists $R$, $U$ such that 
    $U^tMM^tU=R^2$, since $R^2$ has the same shape as $R$ with $t_i$ replaced by $2t_i$. If we set $K=M^tUR^{-1}$, we see that $K\in G$, in particular it has determinant 1. Moreover, $M=KRU$ and $^tKK=R^{-1}U^tMM^tUR^{-1}=1$.
    
    \item Let us prove that the group $\SO_{2N}(\R)\cap \SpONN$ is pathwise connected. This, together with the decomposition given by point (5), will imply that $\SpONN$ is pathwise connected.
To this purpose, we introduce the matrices \[S'=\begin{pmatrix}
I_k&0\\0&-I_{k'}\end{pmatrix}\text{ and }J'=\begin{pmatrix}
0&-I_{k/2}&0&0\\I_{k/2}&0&0&0\\0&0&0&-I_{k'/2}\\0&0&I_{k'/2}&0\end{pmatrix},\]
where $k=k'=N$ when $N$ is even and $k=N+1$ and $k'=N-1$ when $N$ is odd.
    We see that  $\SO_{2N}(\R)\cap \SpONN$ is homeomorphic to  $\SO_{2N}(\R)\cap G'_5$, where we have
    \[G'_5=\left\{M\in\mathcal{M}_{2N}(\R), ^tMJ'M=J', ^t MS'M=S'\right\}.\]

Indeed, the matrices $J$ and $J'$ (resp. $S$ and $S'$) are unitarily equivalent since $J'$ is obtained from $J$ by permutating rows and columns. Therefore, $\SpONN$ and $\SpONN'$ are unitarily equivalent hence homeomorphic.  

We take a matrix $M=\begin{pmatrix} A&B\\C&D\end{pmatrix}\in \SO_{2N}(\R)\cap G'_5$, where the blocks have respective size $k$ and $k'$.
Since $^tMM=I_{2N}$,  \begin{equation}\left\{\begin{array}{r c l}
^tAB+^tCD&=&0\\
^tAA+^tCC&=&I_k\\
^tBB+^tDD&=&I_{k'}
\end{array}\right.\end{equation}
Similarly, since $^tMS'M=S'$,  \[\left\{\begin{array}{r c l}
^tAB-^tCD&=&0\\
^tAA-^tCC&=&I_k\\
^tBB-^tDD&=&-I_{k'}
\end{array}\right.\]
These two series of equations imply that $B=C=0$. We have then that $A$ and $D$ are orthogonal matrices. 
Finally, the fact that $^tMJ'M=J'$ implies that the two matrices $A$ and $D$ are symplectic matrices. As a consequence, all matrix $M\in \SpONN'\cap\SO_{2N(\R)}$ can be written $\begin{pmatrix} A&0\\0&D\end{pmatrix}$, where $A,D\in \mathrm{Sp}_{k^*/2}(\R)\cap \SO_{k^*}(\R),$ $k^*$ being $k$ or $k'$. One easily checks that all such matrices are in $\SpONN'\cap\SO_{2N}(\R)$.
But, for all $k^*$, $\mathrm{Sp}_{k^*/2}(\R)\cap \SO_{k^*}(\R)$ is pathwise connected as an intersection of 2 pathwise connected Lie groups so $\SpONN'\cap\SO_{2N}(\R)$  is itself pathwise connected. For the pathwise connectedness of $\SO_{k^*}(\R)$, see \cite{AW67}. For the pathwise connectedness of the symplectic group, one uses the fact that this group is generated by the symplectic transvections (see \cite{J85}, Lemma~1, p.~392) to construct a continuous path between any symplectic matrix and the identity matrix.

\item By differentiating both relations $^tMJM=J$ and $^tMSM=S$ one gets that the Lie algebra of $\SpONN$ is the set
$$\left\{ M\in \mathcal{M}_{2N}(\R)\ |\ ^tMJ+JM=0 \mbox{ and } ^tMS+SM=0 \right\} \subset \mathfrak{sp}_N(\R) .$$
If $M=\left(\begin{smallmatrix}
A & B \\ C & D \end{smallmatrix}\right)\in \mathcal{M}_{2N}(\R)$ satisfies $^tMJ+JM=0$ one already gets that $M=\left(\begin{smallmatrix}
 A & B \\ C & {-}^tA   
\end{smallmatrix}\right)$ with $A\in \mathcal{M}_{N}(\R) $ $^tB=B$ and $^tC=C$. For such a matrix $M$, we write the relation $^tMS+SM=0$:
$$\left(\begin{array}{cc}
  ^tA   & C \\
   B  & -A
\end{array} \right) \left(\begin{array}{cc}
  K   & 0 \\
   0  & K
\end{array} \right) + \left(\begin{array}{cc}
  K   & 0 \\
   0  & K
\end{array} \right) \left(\begin{array}{cc}
  A   & B \\
   C  & -^tA
\end{array} \right) =0. $$
Since $B$, $C$ and $K$ are symmetric matrices, it implies that $^tAK+KA =0$ and $BK+KC=0$. This is equivalent, by definitions of the matrix $K$ and of the operator $^s$, to $A={^sA}$ and $C={^sB}$. 

\item First, we remark that since the Lie algebra of the group $\SpONN$ is a sub Lie algebra of the Lie algebra of the symplectic group, its Killing form is the restriction of the Killing form of the Lie algebra of the symplectic group to the Lie algebra of $\SpONN$ given by $(X,Y)\mapsto 2(N+1)\mathrm{Tr}(XY)$ on $\mathrm{Lie}(\SpONN) \times \mathrm{Lie}(\SpONN)$ (see \cite{TY05}). Then, a direct computation shows that this Killing form is non-degenerate, hence the Lie algebra of $\SpONN$ is semisimple. Moreover, since the relationships defining $\SpONN$ are polynomial in the matrix coefficients of its elements, the group $\SpONN$ is an algebraic group (actually it is an algebraic subgroup of the symplectic group). Hence using for example \cite[Proposition 27.2.2]{TY05}, the semisimplicity of the Lie algebra of $\SpONN$ imply the semisimplicity of $\SpONN$ as an algebraic group : the derived group of $\SpONN$ is equal to $\SpONN$.  
\end{enumerate}
\end{proof}

\begin{rem}
The semisimplicity of $\SpONN$ implies that it is topologically perfect hence we can apply to $\SpONN$ the result of Breuillard and Gelander, Theorem \ref{thm_breuillard}, as in the case of the symplectic group.
\end{rem}

\end{document}